\newif\iffull
\fulltrue

\documentclass[twoside,leqno,twocolumn]{article}

\usepackage[letterpaper]{geometry}

\usepackage{ltexpprt}
\usepackage{cite}
\usepackage{amsmath,amssymb,amsfonts}
\usepackage{graphicx}
\usepackage{textcomp}
\usepackage{xcolor}

\usepackage[utf8]{inputenc} 
\usepackage[T1]{fontenc}    

\usepackage{graphicx}
\usepackage{fullpage}
\usepackage{hyperref}       
\usepackage{cleveref}

\usepackage{url}            
\usepackage{booktabs}       
\usepackage{nicefrac}       
\usepackage{microtype}      
\PassOptionsToPackage{dvipsnames}{xcolor}
\usepackage{thm-restate}
\usepackage{mathtools}
\usepackage{xspace}
\usepackage{bm}
\usepackage{multirow}
\usepackage{tabularx}
\usepackage{enumitem}
\usepackage{todonotes}
\usepackage{algorithm}
\usepackage[noend]{algpseudocode}
\usepackage{algorithmicx}
\usepackage{subcaption} 
\usepackage[export]{adjustbox} 
\usepackage{thm-restate}

\def\BibTeX{{\rm B\kern-.05em{\sc i\kern-.025em b}\kern-.08em
    T\kern-.1667em\lower.7ex\hbox{E}\kern-.125emX}}
\begin{document}

\title{\dynhac{}: Fully Dynamic Approximate Hierarchical Agglomerative Clustering}

\newcommand{\mysubsection}[1]{\smallskip\noindent {\bf #1}}
\newcommand{\id}[1]{\ifmmode\mathit{#1}\else\textit{#1}\fi}
\newcommand{\const}[1]{\ifmmode\mbox{\textc{#1}}\else\textsc{#1}\fi}
\newenvironment{proofof}[1]{\bigskip \noindent {\bf Proof of #1:} }

\newcommand{\degree}[1]{\ensuremath{deg(#1)}}
\newcommand{\vertexid}{vertex ID}
\newcommand{\dynorient}{dynamic graph-orientation}
\newcommand{\trianglelink}{triangle-based linkage}
\newcommand{\bestedge}[1]{\ensuremath{\textsc{BestEdge}(#1)}}
\newcommand{\bestedgenoarg}{\ensuremath{\textsc{BestEdge}}}
\newcommand{\bestedgetext}{best-edge}
\newcommand{\union}[1]{\ensuremath{\textsc{Union}(#1)}}
\newcommand{\unionnoarg}[1]{\ensuremath{\textsc{Union}}}
\newcommand{\uniontext}[1]{union}
\newcommand{\merge}[1]{\ensuremath{\textsc{Merge}(#1)}}
\newcommand{\mergenoarg}{\ensuremath{\textsc{Merge}}}
\newcommand{\mergetext}{merge}
\newcommand{\heap}[1]{\ensuremath{\textsc{Heap}(#1)}}
\newcommand{\nghheap}[1]{neighbor-heap}
\newcommand{\mergecost}[1]{\ensuremath{\mathsf{Cost}(#1)}}
\newcommand{\totaledges}[1]{\ensuremath{\mathcal{T}(#1)}}
\newcommand{\staleness}[1]{\ensuremath{\mathcal{S}(#1)}}

\newcommand{\argmin}{\mathop{\mathrm{argmin}}\nolimits} 
\newcommand{\argmax}{\mathop{\mathrm{argmax}}\nolimits} 

\newcommand{\STAB}[1]{\begin{tabular}{@{}c@{}}#1\end{tabular}}

\newcommand{\parhac}{$\mathsf{ParHAC}$}
\newcommand{\seqhac}{$\mathsf{SeqHAC}$}
\newcommand{\optrac}{$\mathsf{OptimizedRAC}$}
\newcommand{\rac}{$\mathsf{RAC}$}
\newcommand{\uwalbucketmerge}{\parhac{}-BucketMerge}
\newcommand{\whp}[1]{\emph{whp}}
\newcommand{\sccsim}{SCC$_{\text{sim}}$}

\newcommand{\parhacexact}{\parhac{}$_{\mathcal{E}}$}
\newcommand{\parhacappx}{\parhac{}$_{0.1}$}
\newcommand{\seqhacexact}{\seqhac{}$_{\mathcal{E}}$}
\newcommand{\seqhacappx}{\seqhac{}$_{0.1}$}
\newcommand{\boruvka}{Bor\r{u}vka}

\newcommand{\cunused}{\ensuremath{C_{\mathsf{unused}}}}
\newcommand{\ctrue}{\ensuremath{C_{\mathsf{true}}}}
\newcommand{\cfalse}{\ensuremath{C_{\mathsf{false}}}}
\newcommand{\wunit}{\ensuremath{w_{\mathsf{unit}}}}
\newcommand{\gand}{\ensuremath{\mathsf{AND}}}
\newcommand{\gor}{\ensuremath{\mathsf{OR}}}

\newcommand{\graphheap}{GraphHAC-Heap}
\newcommand{\heapbased}{heap-based}
\newcommand{\graphchain}{GraphHAC-Chain}
\newcommand{\chainbased}{chain-based}

\newcommand{\singlelink}{single-linkage}
\newcommand{\completelink}{complete-linkage}
\newcommand{\avglink}{average-linkage}
\newcommand{\weightedavglink}{weighted average-linkage}
\newcommand{\unweightedavglink}{unweighted average-linkage}

\newcommand{\linkweight}[2]{\ensuremath{\mathcal{L}(#1, #2)}}
\newcommand{\combsim}[2]{\ensuremath{\mathcal{C}(#1, #2)}}

\newcommand{\hacweight}[2]{\ensuremath{\mathcal{W}(#1, #2)}}

\newcommand{\seqheap}{\ensuremath{\mathsf{Seq}\mhyphen\mathsf{Heap}}}
\newcommand{\seqchain}{\ensuremath{\mathsf{Seq}\mhyphen\mathsf{Chain}}}

\newcommand{\storedwgh}[1]{\ensuremath{W_{\mathcal{S}}(#1)}}
\newcommand{\truewgh}[1]{\ensuremath{W_{\mathcal{T}}(#1)}}
\newcommand{\cut}[1]{\ensuremath{\mathsf{Cut}(#1)}}

\newcommand{\unite}[2]{\ensuremath{\textsc{Unite}(#1, #2)}}
\newcommand{\find}[1]{\ensuremath{\textsc{Find}(#1)}}


\newcommand{\ccclass}{\ensuremath{\mathsf{CC}}}
\newcommand{\pclass}{\ensuremath{\mathsf{P}}}
\newcommand{\ncclass}{\ensuremath{\mathsf{NC}}}
\newcommand{\lclass}{\ensuremath{\mathsf{L}}}

\newcommand{\cchard}{\ensuremath{\mathsf{CC}}-hard}
\newcommand{\cccomplete}{\ensuremath{\mathsf{CC}}-complete}
\newcommand{\pcomplete}{\ensuremath{\mathsf{P}}-complete}
\newcommand{\cc}{\ensuremath{\mathsf{CC}}}
\newcommand{\nc}{\ensuremath{\mathsf{NC}}}

\newcommand{\new}[1]{{\color{purple} #1}}

\definecolor{best}{rgb}{0.0, 0.5, 0.0}
\newcommand{\best}[1]{\color{best}{\underline{#1}}}

\definecolor{munsell}{rgb}{0.0, 0.5, 0.69}
\definecolor{burntsienna}{rgb}{0.91, 0.45, 0.32}
\newcommand{\parcolor}[1]{{\color{munsell}{#1}}}

\newcommand{\myparagraph}[1]{\smallskip\noindent {\bf #1.}}

\newtheorem{definition}{Definition}
\crefname{observation}{Observation}{Observations} 
\crefname{claim}{Claim}{Claims}
\newcommand{\defn}[1]{\textbf{\emph{#1}}}

\newcommand{\minmerge}[0]{\ensuremath{\mathrm{M}}}
\newcommand{\wmax}[0]{\ensuremath{w_{\max}}}
\newcommand{\apxwmax}[0]{\ensuremath{\tilde{w}_{\max}}}
\newcommand{\apx}[1]{\ensuremath{\tilde{#1}}}

\newcommand{\terahac}{$\mathsf{TeraHAC}$}
\newcommand{\dynhac}{$\mathsf{DynHAC}$}

\newcommand{\subhac}{$\mathsf{SubgraphHAC}$}
\newcommand{\flatten}{$\mathsf{Flatten}$}
\newcommand{\wthreshold}{\ensuremath{t}}
\newcommand{\gness}{\ensuremath{\mathrm{goodness}}}
\newcommand{\scc}{\ensuremath{\mathsf{SCC}}}

\newcommand{\kuba}[1]{{\color{cyan} Kuba: #1}}
\newcommand{\shangdi}[1]{{\color{blue} Shangdi: #1}}
\newcommand{\laxman}[1]{{\color{purple} Laxman: #1}}
\newcommand{\jason}[1]{{\color{brown} Jason: #1}}
\newcommand{\nikos}[1]{{\color{orange} Nikos: #1}}

\definecolor{forestgreen}{rgb}{0.13, 0.55, 0.13}
\newcommand{\forestgreen}[1]{{\color{forestgreen}{#1}}}
\newcommand{\revision}[1]{{{#1}}}
\newcommand{\shepchange}[1]{{#1}}

\author{
Shangdi Yu\thanks{MIT}
\and 
Laxman Dhulipala\thanks{UMD}
\and 
Jakub {\L}{\k{a}}cki\thanks{Google}
\and 
Nikos Parotsidis\footnotemark[3]
}

\date{}

\maketitle

\fancyfoot[R]{\scriptsize{Copyright \textcopyright\ 2025 by SIAM\\
Unauthorized reproduction of this article is prohibited}}

\begin{abstract}
We consider the problem of maintaining a hierarchical agglomerative clustering (HAC) in the dynamic setting, when the input is subject to point insertions and deletions.
We introduce DynHac -- the first dynamic HAC algorithm for the popular \emph{average-linkage} version of the problem which can maintain a $1+\epsilon$ approximate solution.
Our approach leverages recent structural results on $1+\epsilon$-approximate HAC~\cite{terahac} to carefully identify the part of the clustering dendrogram that needs to be updated in order to produce a solution that is consistent with what a full recomputation from scratch would have output.

We evaluate DynHAC on a number of real-world graphs. We show that DynHAC can handle each update up to 423x faster than what it would take to recompute the clustering from scratch. At the same time it achieves up to 0.21 higher NMI score than the state-of-the-art dynamic hierarchical clustering algorithms, which do not provably approximate HAC.
\end{abstract}

\section{Introduction}\label{sec:intro}
Clustering is an unsupervised machine learning method that has been
widely used in many fields including computational biology, computer
vision, and finance to discover structures in a data
set\iffull
~\cite{irbook,
berkhin2006survey,Aggarwal2013,leibe2006efficient}
\fi
.
To group similar objects at all resolutions, a \emph{hierarchical
  clustering} can be used to produce a tree that represents
clustering results at different scales.
\emph{Hierarchical agglomerative clustering} (HAC) is the most prominent hierarchical clustering algorithm
\iffull
~\cite{murtagh2012algorithms,murtagh2017algorithms,mullner2011modern, mullner2013fastcluster, gronau2007optimal, stefan1996multiple}
\fi
which is particularly well-suited to finding a large number of highly precise clusters\iffull
~\cite{zhao2002evaluation, hua2017mgupgma, kobren2017hierarchical, blundell2013bayesian, culotta2007author, terahac, parhac}
\fi
.
The algorithm takes as input a collection of $n$ points and a function that gives the similarity of each pair of points.
It starts by putting each point in its own singleton clusters and then proceeds in up to $n-1$ steps.
In each step it finds the two most similar clusters, and \emph{merges} them together, that is it replaces them by their union.
The similarity between two clusters is defined by a \emph{linkage function} which maps the all-pair  point-to-point similarities between the points in both clusters to a single similarity value.
One popular similarity function is \emph{average--linkage}, where the similarity between the clusters is the average of all the individual similarities.
That is, for two clusters $C$ and $D$, their similarity is the sum of all-pairs similarities between points in $C$ and $D$ divided by $|C|\cdot |D|$.

The output of the algorithm is a \emph{dendrogram}: a rooted binary tree which describes all merges performed by the algorithm.
Specifically, each node of the dendrogram represents a cluster.
The leaves correspond to singleton clusters representing the input points.
Each internal node of the dendrogram represents the cluster obtained by merging its children.

Since specifying similarities between all $n \times n$ pairs of input points is infeasible for large datasets, recent work on scaling up HAC considered the graph-based version of the problem\iffull
~\cite{dhulipala2021hierarchical, parhac, terahac, monath2021scalable,NIPS2017_2e1b24a6}.
\else
~\cite{terahac}.
\fi
In graph-based HAC, the input is a (typically sparse) similarity graph, whose vertices represent points, and each edge specifies the similarity between its endpoints.
If the input is a metric space, a natural approach is to build the $k$-nearest neighbor graph, that is a graph, where each point is connected to its $k$ most similar other points.
Using graph-based representation and allowing approximation, average linkage HAC has been scaled to datasets of billions of points~\cite{parhac, terahac}.

Contemporary data exists in a constant state of flux, e.g, users interact with platforms in frequent intervals, or financial transactions occur very frequently. This has lead researchers to focus on maintaining solutions to problems like clustering in the \emph{dynamic} setting where nodes and edges are being inserted and deleted and the objective is to adjust the solution efficiently to reflect the new state of the data. This is not an easy task, and studies often settle for algorithms that work in the (restricted) incremental only setting, i.e., where nodes or edges are inserted but not deleted.

Our goal is to maintain a HAC dendrogram in the \emph{fully dynamic} setting, that is updated under a sequence of both insertions and deletions.
Our main focus is to obtain highly precise clustering. Because of this, we use the average-linkage similarity function, which is known to deliver excellent empirical quality
\iffull
~\cite{zhao2002evaluation, hua2017mgupgma, kobren2017hierarchical, blundell2013bayesian, culotta2007author, terahac, parhac}
\else
~\cite{terahac, parhac}
\fi
and strive to obtain rigorous theoretical quality guarantees on the output dendrogram.

Our main contribution is \dynhac{} -- a dynamic HAC algorithm maintaining a $1+\epsilon$ approximate average-linkage dendrogram under point insertions and deletions. We use the notion of approximation introduced by Moseley et al.\iffull
~\cite{48657, dhulipala2021hierarchical,parhac, terahac}.
\else
~\cite{48657}.
\fi
Namely, a $1+\epsilon$ approximate HAC algorithm is allowed to merge two clusters with similarity at most a $1+\epsilon$ factor away from the similarity of the two most similar clusters.

The main challenge in developing an efficient dynamic HAC algorithm is the sensitivity of the output to even small changes in the input.
In particular, one can easily design instances where inserting even a single node or edge causes the resulting dendrogram to change completely.
Interestingly, it was recently shown that even if one allows super-constant approximation, maintaining average-linkage HAC in a dynamic setting requires $n^{\Omega(1)}$ time per update in the worst case~\cite{tseng2022parallel}.
However, the very hard instances for dynamic HAC require a very particular structure, and so they do not preclude dynamic HAC from working efficiently on real-world instances.

Our algorithm builds on the ideas behind \terahac{}~\cite{terahac}, a recently introduced distributed HAC algorithm.
\terahac{} partitions the nodes into disjoint partitions and then independently runs an algorithm called \subhac{} within each partition.
The goal of \subhac{} is to perform a certain number of $1+\epsilon$ approximate HAC steps within the partition.
Crucially, in order to achieve $1+\epsilon$ approximation, \subhac{} uses a carefully chosen stopping condition to ensure that the merges performed within each partition are consistent with what a $1+\epsilon$ HAC algorithm would have performed if it was run on the entire graph.
We note that this condition is based on the nodes in the partition and the set of its incident edges, that is edges that have at least one endpoint in the partition.
Once all such intra-partition merges have been performed,
the graph is partitioned again and the above step is repeated.
We call each step of the above algorithm a \emph{round}.
On real-world datasets the graph typically shrinks by a constant factor in each round, and so the total number of rounds is small.

\dynhac{} uses a similar partitioned approach.
We leverage the fact that if an edge is inserted within a partition $P$, no other partitions within the same level are affected.
For the affected partition $P$, we simply run \subhac{} from scratch.
However, with multiple levels, the problem becomes more intricate.
A single change within the partition $P$ may cause \subhac{} to perform a very different set of merges.
This in turn may cause a nontrivial change to the graph handled in the next level.
Moreover, in addition to having to propagate changes to the graph, we also need to dynamically maintain the partitions to ensure that their size is balanced and ensure that \subhac{} runs efficiently.
Hence, on each level we need to carefully decide which partitions to recompute based on the changes of the previous level, and the updates to the partitioning that are taking place.

We experimentally evaluate \dynhac{} and compare it to state of the art dynamic hierarchical clustering algorithms.
Compared to static algorithms, we observe that \dynhac{} delivers up to a 423x speedup over recomputing from scratch after each update.
Compared to existing dynamic hierarchical clustering algorithms, none of which provably approximate HAC, we observe that \dynhac{} achieves up to 0.21 higher NMI score, showcasing the value of provable approximation guarantees in practice.

Our code, data, and a full version of our paper can be found at \url{https://github.com/yushangdi/dynamic-hac}.

\subsection{Related Work}

Several studies have considered the problem of dynamically maintaining a hierarchical clustering \cite{menon2019online, zhang1996birch, garg2006pbirch, kobren2017hierarchical, monath2019scalable, monath2023online, charikar2019hierarchical, cohen2019hierarchical, dasgupta2016cost, moseley2023approximation, rajagopalan2021hierarchical, vainstein2021hierarchical}. 
However, the more efficient algorithms can only process incremental updates. We further expand on related work in the appendinx.

The work that is most similar to ours is \cite{monath2023online}. The authors consider multiple incremental clustering approaches both top-down and bottom-up. 
Their approaches include 1) modifications of the top-down Stable Greedy (SG) Trees approach by~\cite{zaheer2019terrapattern} \iffull that also allow re-evaluating the greedy choices in a second phase\fi, and 
2) dynamic versions of bottom-up approaches like the RecipNN method~\cite{sumengen2021scaling} which is an efficient implementation of HAC, and Affinity clustering~\cite{bateni2017affinity}. 
The authors show that the bottom-up exhibit better tradeoffs between quality and running time.
While their Online RecipNN algorithm from \cite{monath2023online} is supposed to be a dynamic version of exact HAC, some simplifications in the implementation make it diverge from the exact algorithm.
As a result, their experimental analysis concluded that the method often produces inferior quality results compared to other methods, and its running time is not competitive.
By maintaining a $(1+\epsilon)$-approximate HAC dendrogram, we improve the performance of the algorithm by allowing some approximation, and at the same time consistently maintain a high quality clustering.
Similarly to the case of static algorithms~\cite{parhac, terahac}, we observe that leveraging approximations leads to significant speedups.

\section{Preliminaries}

Let $G=(V, E, w)$ be a weighted and undirected graph, where $|V| = n$ and $w$ is the edge-weight function.
We assume all edge weights are positive.
$V(G)$ and $E(G)$ denote the vertex and edge sets of $G$, respectively.

In order to define $(1+\epsilon)$-approximate HAC we describe a sequential static $(1+\epsilon)$ approximate HAC algorithm, which we refer to as \seqhac.
\seqhac{} maintains a graph, where each vertex is a cluster.
Initially, each vertex is a singleton cluster, and so the state is represented by the input graph $G$.
In addition, it also maintains the \emph{size} of each vertex of $G$ (which is initially $1$).
We define the \emph{normalized} weight of an edge $xy$ as $\bar{w}(xy) = w(xy) / (S(x) \cdot S(y))$, where $S(x)$ and $S(y)$ are the sizes of $x$ and $y$ respectively.

\seqhac{} proceeds as follows.
While $G$ has at least one edge, we pick an edge $xy$ whose normalized weight satisfies $\bar{w}(xy) \geq \bar{w}(uv) / (1+\epsilon)$, where $uv$ is the edge with the highest normalized weight in $G$.
Then, the algorithm \emph{contracts} the edge $xy$. 
We refer to this event as a \emph{merge} of $x$ and $y$ of linkage similarity $\bar{w}(xy)$.
Contraction of $xy$ merges $x$ and $y$ into one vertex $z$ with size $S(z) = S(x) + S(y)$.
The parallel edges that are created are merged into one, and the corresponding edge weights are summed.
Finally, we remove self-loops.
Notice that when $\epsilon = 0$, \seqhac{} is equivalent to the algorithm sketched in \cref{sec:intro}.

Observe that \seqhac{} can produce multiple valid outputs, given that it can contract any edge of sufficiently high weight in each step.
We say that any valid output of \seqhac{} is a $(1+\epsilon)$-approximate dendrogram.
Moreover, any algorithm which always produces a $(1+\epsilon)$-approximate dendrogram is called $(1+\epsilon)$-approximate HAC.

\boldsymbol{\terahac{}} \textbf{Algorithm.}
\terahac{}~\cite{terahac} is a {\em parallel} $(1+\epsilon)$ approximate HAC algorithm.
The \dynhac{} algorithm that we introduce in this paper is essentially a dynamic version of \terahac{}, and so we now briefly outline how \terahac{} works.
\terahac{} proceeds in multiple rounds.
In each round it computes a partition $\mathcal{P}$ of the vertices of $G$, i.e., $\bigcup_{p \in \mathcal{P}} p = V(G)$, and contracts some edges within each partition.

\begin{definition}[Partition subgraph]\label{def:ps}
Let $\mathcal{P}$ be a partition of $V$.
For each $p \in \mathcal{P}$ we define the \emph{partition subgraph} of $p$, denoted by $H_p$, as follows.
The vertex set of $H_p$ is the set of vertices $p$ and all their neighbors in $G$.
The edge set of $H_p$ is the set of all edges incident to a vertex in $p$.
For each $H_p$, we say that the vertices of $p$ are \emph{active}, and the remaining ones are \emph{inactive}.
\end{definition}

Observe that any edge $xy$ of $G$ is in at most $2$ subgraphs $H_p$.
Moreover, the number of vertices in $V(H_p) \setminus p$ is at most the number of edges in $H_p$.
As a result, all subgraphs $H_p$ (for $p \in \mathcal{P}$) have $O(|E(G)|)$ vertices and edges in total.

\terahac{} then runs a restricted HAC algorithm on each partition subgraph.
We call this algorithm \subhac{}.
\subhac{}, given a partition subgraph $H_p$, merges some pairs of vertices in $H_p$ following two important constraints.
First, it only merges active vertices of $H_p$, which ensures that all \subhac{} calls within a round  merge disjoint sets of vertices.
The vertex obtained by merging two active vertices is considered active as well.
Second, all merges made by \subhac{} are provably consistent with what a $(1+\epsilon)$ approximate HAC algorithm would do if it was working on the entire graph.
To this end, \subhac{} leverages the following notion to decide whether a certain pair of vertices can be merged.

\begin{definition}[Good merge\cite{terahac}]\label{def:good}
Let $\epsilon \geq 0$ and $G_i$ be a graph obtained from $G$ by performing some sequence of merges.
For each vertex $v$ of $G_i$, we define $\minmerge(v)$ to be the smallest linkage similarity among all merges performed to create vertex (cluster) $v$.
Specifically, for each vertex $v$ of size $1$ we have $\minmerge(v) = \infty$, and when two vertices $u$ and $v$ merge to create a vertex $z$, we have $\minmerge(z) = \min(\minmerge(u), \minmerge(v), \bar{w}(uv))$.
Moreover, for each vertex $v$ we use $\wmax(v)$ to denote the highest normalized weight of any edge incident to $v$.
With this notation, we say that a merge of an edge $uv$ in $G_i$ is $(1+\epsilon)-$\emph{good} if and only if
\[
\frac{\max(\wmax(u), \wmax(v))}{\min(\minmerge(u), \minmerge(v), \bar{w}(uv))} \leq 1+\epsilon.
\] 
\end{definition}

\subhac{} ensures that each merge that it produces is $(1+\epsilon)$ good and leverages the following key property shown in~\cite{terahac}.

\begin{lemma}\label{lemma:dendrogram}
Any dendrogram produced by a sequence of $(1+\epsilon)$-good merges is $(1+\epsilon)$ approximate.
\end{lemma}

The \dynhac{} algorithm uses \subhac{} in a black-box way.
We note that in addition to the current graph \subhac{} must also be provided the minimum linkage similarities $M(\cdot)$ of any vertex.

\begin{table}[ht]
\centering
\begin{tabular}{c l}
\hline
Symbol & Meaning \\ \hline
$V_i$ & vertices to insert to $G_i$ in round $i$. \\
$V^d_i$ & vertices to delete from $G_i$ in round $i$.\\
$N_G(V)$ & The union of neighbors of $v \in V$. \\
$P_i$ & Partition of vertices in $G_i$.\\
$\texttt{VMap}_i$ & Maps $V(G_i)$ to the vertices they\\
& contracted to in $G_{i+1}$.\\
$\minmerge$ & the smallest linkage similarity used\\ & to create a cluster.\\
\hline
\vspace{-0.8cm}
\end{tabular}
\end{table}
\section{\dynhac{} Algorithm}
In this section we describe the \dynhac{} algorithm, which is a dynamic version of \terahac{}.
Following \terahac{}, \dynhac{} proceeds in rounds.
In each round it partitions the graph, runs \subhac{} on each partition subgraph and then obtains the input to the next round by applying all the merges from all \subhac{} calls.
We denote by $G_i$ the graph which is the input to round $i$. Here, $G_1$ is the input graph on which we run HAC.

The main principle behind \dynhac{} is as follows. Whenever a partition subgraph changes, we rerun \subhac{} on this subgraph.
Since the only input to \subhac{} is a partition subgraph (together with the corresponding minimum merge similarities), we do not have to rerun \subhac{} on the partition subgraphs that have not changed.
We note, however, that the partition subgraph of a partition $p$ can contain an edge $xy$, such that $y \not\in p$.
Since the normalized weight of the edge $xy$ depends on the size of $y$, the partition subgraph $H_p$ can change as a result of a change to a vertex outside of $p$.

\textbf{Partitioning.}
Because \dynhac{} is dynamic, the partitioning used in each round needs to be updated dynamically.
Moreover, as we rerun \subhac{} from scratch upon a change to a partition subgraph, we want the partitions to be relatively small.
We use the following simple partitioning scheme.
We first randomly assigning the colors red and blue to each vertex with equal probability.
Then, each partition consists of the vertices that have the same partition id, according to the following definition.

\begin{definition}[Partition id]\label{def:pid}
We define the \emph{partition id} of each vertex $v$ as follows.
If $v$ is either red or does not have a blue neighbor, the partition id is $v$.
Otherwise, i.e., when $v$ is blue and has a red neighbor, the partition id of $v$ is the id of its highest normalized weight red neighbor.
\end{definition}

We note that the choice of the partitioning algorithm only affects the running time of the entire clustering algorithm (and not its correctness). 

\textbf{Data to dynamically maintain.}
The \dynhac{} algorithm dynamically maintains the following:
\vspace{-0.15cm}
\begin{itemize}
    \item For each round i:
    \vspace{-0.15cm}
    \begin{itemize} \item $G_i$: the input graph to this round,
    \item  $P_i$: partition of the vertices in $G_i$,
    \item $\texttt{VMap}_i$: a map from vertices in $G_i$ to the corresponding vertices in $G_{i+1}$. If a vertex is not contracted in round $i$, it maps to itself.
    \end{itemize}
    \vspace{-0.2cm}
    \item $\minmerge$: minimum merge similarity of each vertex in each $G_i$ (Definition~\ref{def:good}),
    \vspace{-0.15cm}
    \item $\mathcal{D}$: the $(1+\epsilon)$-approximate dendrogram.
\end{itemize}

\vspace{-0.15cm}
\textbf{Handling an update.}
The algorithm for handling an update is given as \Cref{alg:dynhac}.
It takes as input the vertices and edges to insert $V, E$,
the vertices to delete $V^d$, the existing dendrogram $\mathcal{D}$, and a weight threshold $\wthreshold > 0$. The weight threshold $\wthreshold$ is used to terminate the algorithm once it performs all merges of sufficiently high similarity.
For simplicity, we assume that $t$ is the same across all updates, i.e., we maintain the dendrogram up to some linkage similarity.
The effect of this parameter will be discussed later. 

Handling an update starts by adding leaf nodes to the dendrogram and assigning $\infty$ to $\minmerge(v)$ for the newly inserted nodes (Lines~\ref{alg:dynhac:addleaves}--\ref{alg:dynhac:init_min_merges}).
Then, the algorithm updates each round one by one (Line~\ref{alg:dynhac:for}).
In each round it updates the partitions of the graph, and runs \subhac{} on each affected partition subgraph. 
Thus, if the update in a round is not incident to a partition $p$, the merges in $p$ are all still $(1+\epsilon)$-good and $p$ does not need to be re-clustered in this round.

\begin{algorithm}[t]
\caption{\dynhac{}-$\epsilon$}\label{alg:dynhac}
\begin{algorithmic}[1]
\State \textbf{Input:} $V, E, V^d$, $\wthreshold \geq 0$
\State \textbf{Update:} $\mathcal{D}$, $\minmerge$, $\{G_i, P_i, \texttt{VMap}_i\}$
\State $\mathcal{D}$.AddLeaves($V$) \Comment{Add leaf nodes to $\mathcal{D}$} \label{alg:dynhac:addleaves}
\State $\minmerge(v) = \infty$  for $v \in V$ 
\label{alg:dynhac:init_min_merges}
\State $V_1, E_1, V^d_1 = V, E, V^d$  \label{alg:dynhac:init_change}
\For{$i \in \{1 \dots \infty\}$} \label{alg:dynhac:for}
\State $V_{i+1}, E_{i+1}, V^d_{i+1} = $ DynHacRound($V_i$, $E_i$, $V^d_i$)\label{alg:dynhac:round}
\If{$G_i$ has no edge with weight $> \wthreshold / (1+\epsilon$)}\label{alg:dynhac:stop}
\State Empty rounds $\{i +1, \dots, \infty \}$
\State Remove the ancestors of $V(G_i) \cup V^d$ in $\mathcal{D}$
\State $\mathcal{D}$.RemoveLeaves($V^d$)
\State \textbf{break};
\EndIf
\EndFor
\end{algorithmic}
\end{algorithm}

\textbf{Updating a round.}
The process of updating a single round is shown as Algorithm~\ref{alg:dynhacround}.
For each round, we 1) update $G_i$, $\texttt{VMap}_i$, $P_i$, $\minmerge$ and $\mathcal{D}$ according to the inserted/deleted vertices and edges, and 2) compute the new vertices and edges ($V_{i+1}, E_{i+1}$) to add to and vertices ($V^d_{i+1}$) to delete from the next round to satisfy the invariant. The computation only depends on the state of the previous round. 

When a graph $G_i$ has no edge of weight $> \wthreshold / (1+\epsilon)$ (Line~\ref{alg:dynhac:stop}), no more merges are needed, so we cleanup and stop. To cleanup, we empty all future rounds and remove the dendrogram ancestors of deleted nodes $V^d$ and vertices in the last graph $V(G_i)$. We also remove the deleted leaf nodes in the dendrogram.

\begin{algorithm}[t]
\caption{\dynhac{} Round}\label{alg:dynhacround}
\begin{algorithmic}[1]
\State \textbf{Input:} $V_i$, $E_i$, $V^d_i$
\State \textbf{Output:} $V_{i+1}, E_{i+1}$, $V^d_{i+1}$ 
\State \textbf{Update:} $G_i$, $P_i$, $\minmerge$, $\texttt{VMap}_i$, $\mathcal{D}$
\State $\Delta_P \gets$ UpdatePartition($G_i$, $V_i$, $E_i$, $V^d_i$, $P_i$) \label{alg:dynhacround:updatepartition}
\State $G_i \gets G_i \cup (V_i, E_i) \setminus V^d_i$ \Comment{Update graph} \label{alg:dynhacround:updategraph}
\If{$V_{i+1}$ is empty}  \Comment{Reached last round} \label{alg:dynhacround:if}
\State $P_{\text{dirty}} \gets $ all partitions 
\Else
\State $P_{\text{dirty}} \gets $ DirtyPartitions($\Delta_P$, $G_i$)  \label{alg:dynhacround:dirtypartitions}
\EndIf

\State $V_{i+1}, E_{i+1}, V^d_{i+1} = \{\}$
\For{$p \in P_{\text{dirty}}$} 
    \State $H_p \gets $ Subgraph($G_i$, $p$) \label{alg:dynhacround:subgraph}
    \State $\texttt{merges}, \mathcal{D}_p, H_p^{\text{c}} \gets $ SubgraphHAC($H_p$, $\minmerge$, $\epsilon$)  \label{alg:dynhacround:subgraphhac}
    \State UpdateDendrogram(\texttt{merges}, $\mathcal{D}$) \label{alg:dynhacround:updated}
    \State UpdateMinMergeSim(\texttt{merges}, $\minmerge$) \label{alg:dynhacround:updatem}
    \State $V^d_{i+1}\gets V^d_{i+1} \cup $ UpdateVMap($\mathcal{D}_p$, $p$,  $V^d_i$, $\texttt{VMap}_i$)\label{alg:dynhacround:updatevmap}
    \State $H_p^{\text{c}} \gets $ contract inactive vertices in $H_p^{\text{c}}$ based on $\texttt{VMap}_i$. \label{alg:dynhacround:contract}
    \State $V_{i+1} \gets V_{i+1} \cup V_{\text{active}}(H^{\text{c}}_p) \setminus V(G_{i+1})$ \label{alg:dynhacround:newv}
    \State $E_{i+1} \gets E_{i+1} \cup  E(H_p^{\text{c}})$ \label{alg:dynhacround:newe}
\EndFor
\State Return $(V_{i+1}, E_{i+1})$, $V^d_{i+1}$
\end{algorithmic}
\end{algorithm}


We now describe the order of updating the state, and then explain the details for updating each object.
\begin{enumerate}
    \vspace{-0.15cm}
    \item We first update the partitioning $P_i$, and obtain the vertices that changed partition ids due to insertions and deletions, along with their partition ids before and after the update, $\Delta_P$ (Line~\ref{alg:dynhacround:updatepartition}). 
    \vspace{-0.15cm}
    \item In Line~\ref{alg:dynhacround:updategraph} we update the graph in this round. When a vertex is removed, all its neighboring edges are removed as well.
    \vspace{-0.15cm}
    \item Lines~\ref{alg:dynhacround:if}--\ref{alg:dynhacround:dirtypartitions} compute the partitions that we rerun \subhac{} on, which we call the \textit{dirty partitions} ($P_{\text{dirty}}$). If the graph in the next round has no vertex, then we've reached the last round, but clustering has not finished, so all partitions are dirty and we need to cluster all of them. Otherwise, we find the dirty partitions based on $\Delta_P$.
    \vspace{-0.6cm}
    \item For each dirty partition $p$, we consider the partition subgraph $H_p$ (see \cref{def:ps}).
    We run the $(1+\epsilon)$ \subhac{} algorithm~ \cite{terahac}, which merges some pairs of active vertices. (Lines~\ref{alg:dynhacround:subgraph}--\ref{alg:dynhacround:subgraphhac}). As a result of running \subhac{}, we obtain a merge sequence $\texttt{merges}$, the resulting dendrogram $\mathcal{D}_p$, and $H_p^{\text{c}}$, which is $H_p$ after applying the merges performed by \subhac{}.
    \vspace{-0.5cm}
    \item In Lines~\ref{alg:dynhacround:updated}--\ref{alg:dynhacround:updatem}, we update the overall dendrogram $\mathcal{D}$ and the minimum similarities $\minmerge$ based on the obtained new merges.
    \vspace{-0.15cm}
    \item In Line~\ref{alg:dynhacround:updatevmap}, we update $\texttt{VMap}_i$ and compute the vertices to delete from the next round.
    \vspace{-0.15cm}
    \item Finally, in Lines~\ref{alg:dynhacround:contract}--\ref{alg:dynhacround:newe}, we obtain new vertices to insert to the next round from the contracted graph. Some vertices may already exist in the next round (we made the same merges as before), these vertices are excluded. The new edges to insert are also obtained from the contracted graph. Intuitively, these are the vertices and edges to insert because $G_{i+1}$ is just the graph $G_i$ contracted according to merges in round $i$.
\end{enumerate}

\subsection{Dynamic update of dendrogram and auxiliary data.}

\textbf{Update partitioning $P_i$.}
We observe that in order to update the partition ids to restore the property of~\cref{def:pid} \textit{we just need to re-compute the partition ids of the new vertices, the neighbors of new vertices, and the neighbors of deleted vertices} (see \cref{lemma:changepartition}). After we compute their new partition ids, we return these vertices $V_{\text{new and neigh}}$ and their partition ids before and after the update. Note that this is a superset of vertices that actually changed partition ids, since some of the partition ids might not change. We still return them because they are useful when computing the dirty partitions.

To compute the new partition id of a blue vertex, we just need to scan the neighbors of the vertex. 
For the neighbors $i$ of new vertices $j$, we can optimize this step to simply check if the new edge $(i, j)$ has a higher normalized weight than the that of the edge between $i$ and its previous heaviest red neighbor.

\begin{restatable}{lemma}{changepartition}\label{lemma:changepartition}
Consider a graph $G$ and a graph $G' = G \cup (V, E) \setminus V^d$, which is obtained from $G$ by adding a set of vertices $V$ and edges $E$, and deleting vertices of $V^d$.
Assume that each edge of $E$ is incident to a vertex in $V$ and not incident to any vertex in $V^d$.
Let $U \subseteq V(G')$ be the set of vertices of $G'$ which have different partition ids in $G$ and $G'$.
Then, $U \subseteq V \cup N_{G'}(V) \cup (N_{G}(V^d) \setminus V^d)$.
\end{restatable}



\textbf{Identify dirty partitions.}
In each round of the update our algorithm identifies a set of \emph{dirty partitions}, that is a set of partitions for which it reruns \subhac{} from scratch.
At a high level, we would like to identify all partitions in which, after the update, some merge performed in the current dendrogram is no longer $(1+\epsilon)$-good.
We call these partitions \emph{truly dirty}.
For correctness, we show that the set of dirty partitions identified by our algorithm is a superset of the set of truly dirty partitions.

Efficiently identifying the dirty partitions is challenging because there are multiple cases to consider.
In particular, a partition is considered dirty when:
\begin{itemize}
\vspace{-0.15cm}
\item a vertex is added to/deleted from the partition, and this can be caused by a vertex addition or deletion, or by a partition id change,
\vspace{-0.15cm}
\item the set of edges leaving a partition changes,
\vspace{-0.15cm}
\item a new partition is introduced.
\end{itemize}


\begin{definition}\label{def:delta_p}
Let $\Delta_P$ contain the partition id (see \cref{def:pid}) before and after the update of the new vertices, the deleted vertices, the neighbors of new vertices, and the neighbors of deleted vertices. 
\end{definition}

Our simple algorithm for identifying dirty partitions based on $\Delta_P$ is in Algorithm~\ref{alg:dirtypartition} in Appendix.

Given $V_{\text{new and neigh}}$ and their partition ids before and after the update, all new partitions are marked as dirty. For old partitions, we mark them as dirty if it is still in the new graph (not deleted) and it is not the same as a blue vertex. Note that each partition id is also a vertex id.

\textbf{Update dendrogram.} 
In Algorithm~\ref{alg:updatedendrogram}, we show how to update the dendrogram given the merges performed by \subhac{} in one partition subgraph. The input \texttt{merges} should be a sequence of merges that are consistent with $\mathcal{D}_p$ -- the dendrogram of merges produced for that subgraph. For each merge $(u,v)$ with parent node id $a$, we check if the merge $(u,v)$ also exists in $\mathcal{D}$. If it is in the dendrogram, we do not need to update anything. If it is not in the dendrogram, we remove the ancestors of $u$ and $v$, and merge them to have parent $a$.
Notice that after each round, $\mathcal{D}$ contains all merges in the previous rounds. 

\begin{algorithm}
\caption{UpdateDendrogram}\label{alg:updatedendrogram}
\begin{algorithmic}[1]
\State \textbf{Input:} \texttt{merges}, $\mathcal{D}$
\State \textbf{Output:} update $\mathcal{D}$
\For{$((u,v) \to a) \in \texttt{merges}$} \Comment{$u,v$ merge to form $a$}
\If{$(u,v) \notin \mathcal{D}$}
\State $ \mathcal{D}$.DeleteAncestors($\{u, v\}$)
\State $\mathcal{D}$.Merge($\{u, v\}$, $a$) \Comment{merge $u,v$ creating $a$}
\EndIf
\EndFor
\end{algorithmic}
\end{algorithm}

\textbf{Updating the minimum merge similarities.} 
We compute the new minimum merge similarities of vertices created by $\texttt{merges}$ and store them in $M$. For space efficiency, we also remove the deleted vertices $V^d_{i}$ and the deleted ancestors in $\mathcal{D}$ in the previous step from $M$, but this is not required for correctness.

\textbf{Update vertex mapping $\texttt{VMap}_i$ and compute vertices to delete from next round.} 
To update $\texttt{VMap}_i$, we only need to update the mappings for $p$ and $V^d_i$, where $p$ is the set of the active vertices in $H_p$ (not contracted). This is because only the mappings of these vertices can possibly change. 

For each deleted vertex $v$, its mapped vertex $\texttt{VMap}(v)$ should be deleted in the next round. We can also remove $v$ from $\texttt{VMap}$. One exception is if $\texttt{VMap}(v)$ is a contracted vertex in $G_p$, then it should not be deleted. This can happen in the following example case. $u,v$ merges to $w$ in round $i$. Originally vertex $u$ maps to $w$, but now vertex $w$ is in round $i$ due to updates, but it maps to itself $w$ in the next round because it does not merge in this round.

For each vertex $v$ in the partition, let its root node in $\mathcal{D}_p$
be $r$. $r$ is also the node that $v$ is contracted to after all merges in this round. So $v$ should be mapped to $r$ in the next round. If its current mapping is not $r$, we update it to map to $r$. In addition, we should remove its current mapping from the next round.

\iffull
We present the pseudocode in Appendix (Alg \ref{alg:updatevmap}).
\fi





\section{Empirical Evaluation}\label{sec:experiments}
We run all experiments on a \texttt{c2-standard-60} Google Cloud instance, consisting of 30 cores (with two-way hyper-threading), with 3.1GHz Intel Xeon Scalable processors and 240GB of main memory.
We made no use of parallelization, i.e., used a single core.

\textbf{Datasets.}
In our experiments, we use the same three datasets as \cite{monath2019scalable}.
\defn{MNIST} contains 28x28 grayscale images of handwritten digits (0-9), embedded to 2 dimensions using UMAP~\cite{SMG2020}. \defn{ALOI}~\cite{aloi} contains 1,000 objects recorded under various imaging circumstances --  108k points in total in 128 dimensions. \defn{ILSVRC\_SMALL} (ILSVRC)~\cite{imagenet} is a 50K subset of the Imagenet ILSVRC12 dataset, embedded into 2048 dimensions using Inception~\cite{inception}. 
 


\textbf{Algorithms.}
We compare \dynhac{} with a static approximate graph HAC and a dynamic hierarchical clustering baseline.

\begin{list}{\textbullet}{%
    \setlength{\leftmargin}{0.5em}
    \setlength{\itemindent}{0em}
    \setlength{\labelwidth}{\itemindent}
    \setlength{\labelsep}{0.5em}
    \setlength{\listparindent}{1em}
    \setlength{\itemsep}{0em}
    \setlength{\parsep}{0em}
    \setlength{\topsep}{0em}
    \setlength{\partopsep}{0em}
}
\item \boldsymbol{\dynhac{}}: Our dynamic approximate graph HAC algorithm implemented in C++ that supports both insertion and deletion. We use \dynhac$_{\epsilon=x}$ to denote \dynhac{} run with parameters $\epsilon=x$. When not
specified, \dynhac{}
denotes running the algorithm with $\epsilon=0.1$. We use threshold $t=0.0001$ for \defn{MNIST} and $t=0.01$ for \defn{ALOI} and \defn{ILSVRC\_SMALL}.

\item \defn{Static HAC}~\cite{parhac}: A static approximate graph HAC algorithm implemented in C++. We run it with $\epsilon=0.1$ and the same thresholds as \dynhac{}.
Since this is a static algorithm, it re-computes the clustering from scratch upon each update.

\item \defn{GRINCH} 
\iffull
~\cite{monath2019scalable, prob_cbr, grinchcode}.
\else
~\cite{monath2019scalable}.
\fi
A dynamic algorithm working directly on points. It is implemented in Python. GRINCH inserts and deletes one point at a time. 
A point can be deleted
from a hierarchy by simply removing the corresponding leaf node. 

\item \defn{GraphGrove} (Grove)
\iffull
~\cite{monath2023online,graphgrove}. 
\else 
~\cite{monath2023online}. 
\fi
Monath et al.~\cite{monath2023online}'s OnlineSCC algorithm on vector data. 
We add a uniformly random noise between $10^{-6}$ and $-10^{-6}$ to each coordinate because it does not support duplicate points. 
We use a maximum of 50 levels with a geometric progression of thresholds from 1 to $10^{-8}$. It supports insertion but not deletion.

\end{list}

\begin{figure}[t]
    \centering
    \begin{subfigure}{\columnwidth}
    \centering
        \includegraphics[width=0.7\textwidth]{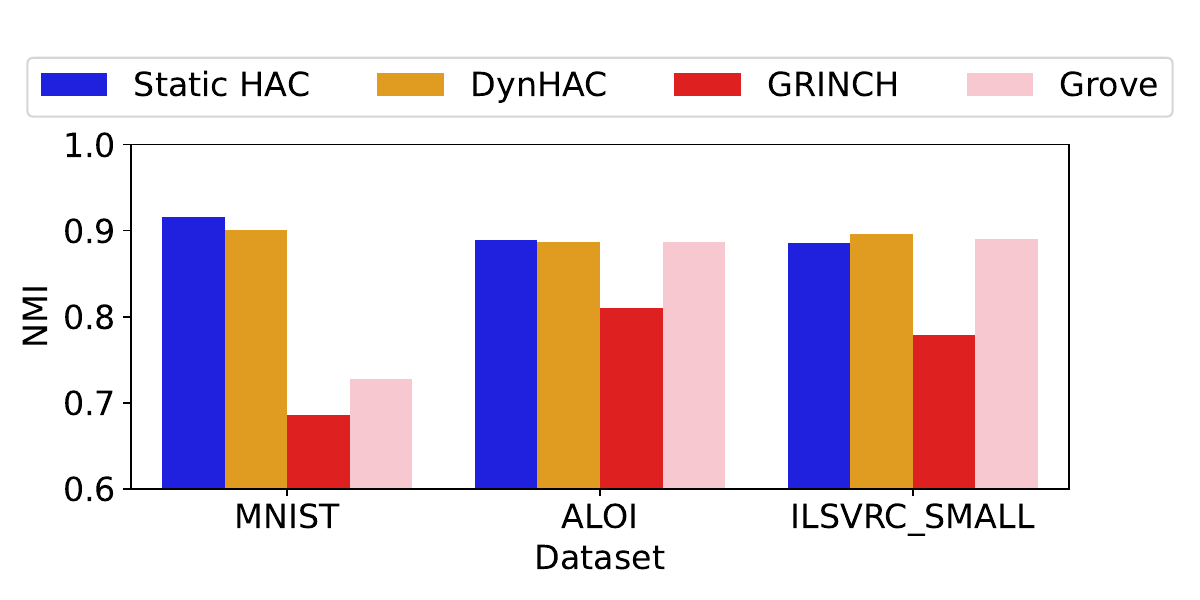}
        \vspace{-0.15cm}
        \caption{\small The NMI of the clustering after all points are inserted. 
        } \label{fig:nmi_bar}
    \end{subfigure}
    
    \begin{subfigure}{\columnwidth}
    \centering
        \includegraphics[width=0.7\textwidth]{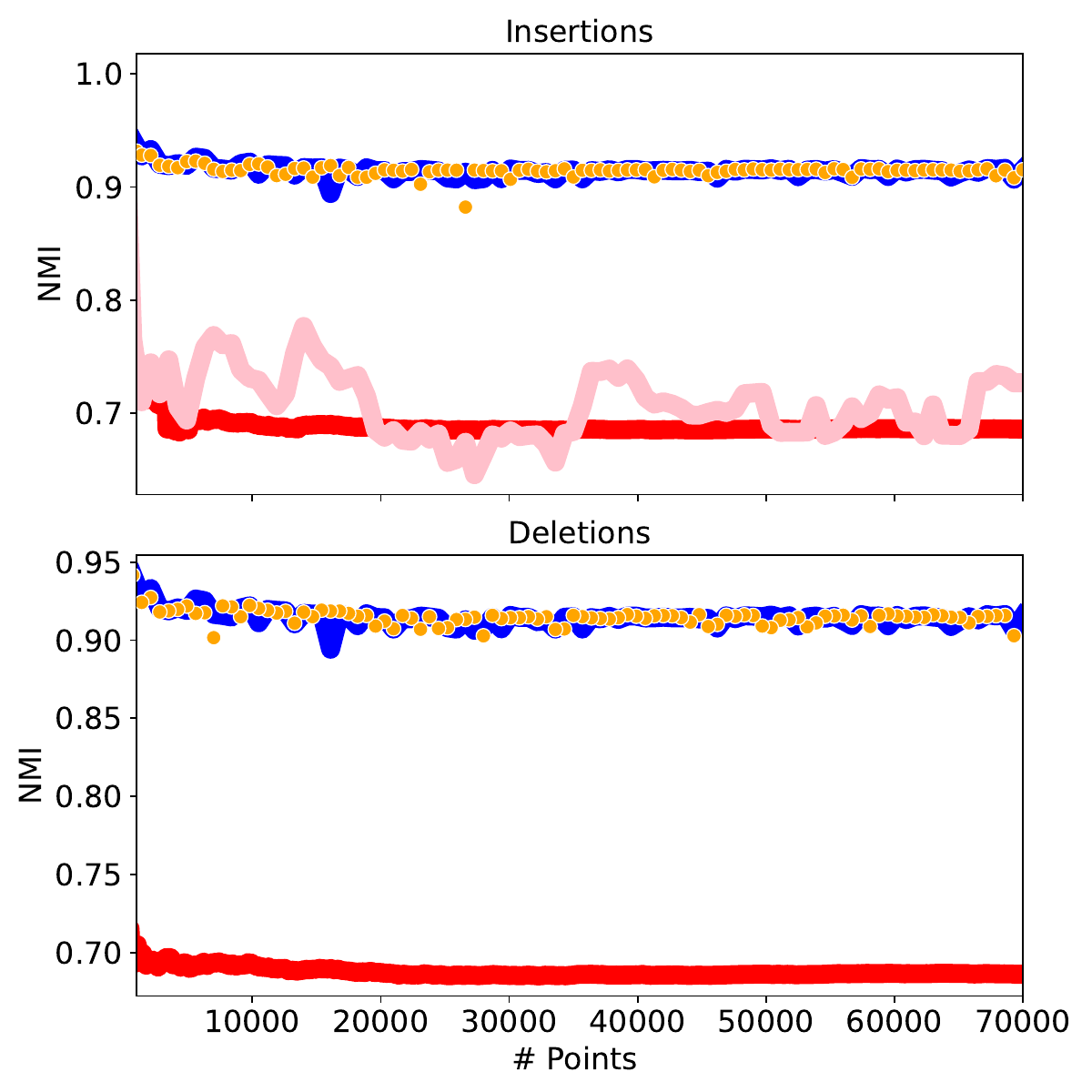}
        \vspace{-0.15cm}
        \caption{\small The NMI after each update on MNIST. $x$-axis is the number of points in the data set after the update.}\label{fig:nmi_mnist}
    \end{subfigure}
        \vspace{-0.55cm}
    \caption{Quality of clustering algorithm.}
    \label{fig:ari}
    \vspace{-3pt}
    
\end{figure}

\begin{figure*}[h]
    \centering
    \begin{subfigure}{\columnwidth}
     \centering
        \includegraphics[width=0.75\textwidth]{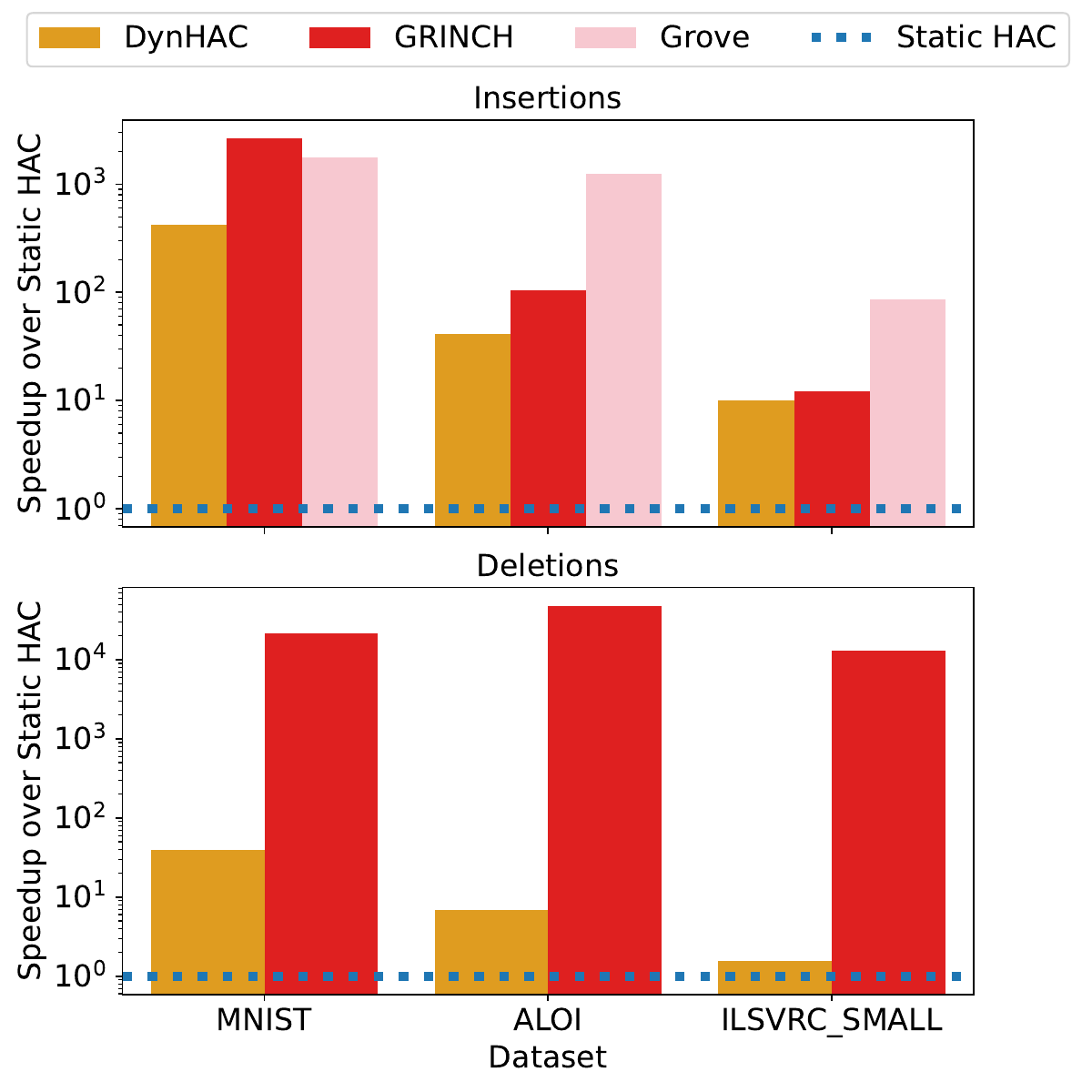}
        \vspace{-0.15cm}
        \caption{\small For \dynhac{} and GRINCH, $y$-axis the averaged running time over the last (insertion) and first (deletion) 100 updates. For the static HAC, $y$-axis the running time of clustering the entire data set.} 
    \end{subfigure}\label{fig:time_bar}
    \begin{subfigure}{\columnwidth}
    \centering
        \includegraphics[width=0.75\textwidth]{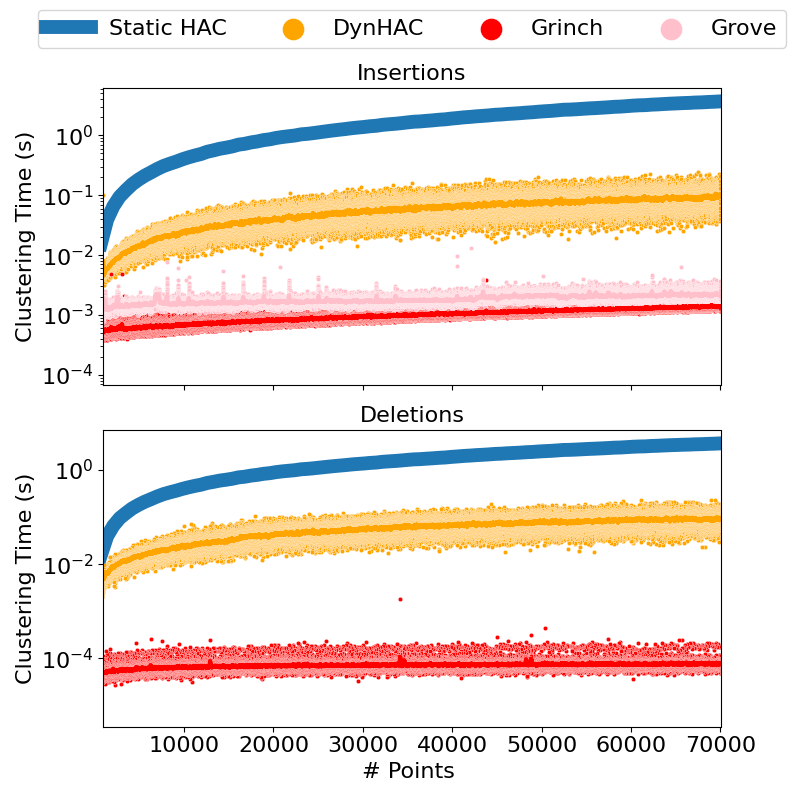}
        \vspace{-0.15cm}
        \caption{\small $x$-axis: the number of points after the update. $y$-axis: time to handle an update (or re-compute from scratch in case of static HAC). The solid line is a running window average of the running time with window size of 100.  } \label{fig:time_mnist}
    \end{subfigure}
        \vspace{-0.25cm}
    \caption{Running times.}
    \label{fig:time}
    \vspace{-17pt}
\end{figure*}

\textbf{Experiment Setup.} 
We get a random permutation of the points $[x_1 \dots x_i \dots x_n]$. We insert points in increasing order of indices $i$, and delete in decreasing order of $i$. We insert/delete one point at time. 

For our insertion experiments with \dynhac{}, we insert a new node into the graph along with edges to its 50 approximately nearest neighbors using the Vamana
\iffull
~\cite{diskann} 
\fi
algorithm from ParlayANN~\cite{dobson2023scaling}. 
We only consider neighbors that are inserted before the point. We first batch insert 99\% of the points, and then one point at a time. We also run an experiment on MNIST where we first insert 1,000 points ($1.42\%$ of all points),
and then one point one at a time.

For our deletion experiments with \dynhac{}, we first batch insert all points, and then remove points one at a time. To batch insert all points, we construct the graph by splitting the points into 100 batches $B_1,\dots, B_{100}$ and for all points in batch $B_i$ we add edges to 50 approximate nearest neighbors in batches $B_1,\dots, B_i$. We use this approach instead of finding the nearest neighbors considering all points to prevent each point  loosing too many neighbors during the deletion sequence, i.e. we ensure that each point has many neighbors that are deleted after it.

For Static HAC, we use the same construction with 100 batches, and run static HAC on the graph.

For GRINCH insertion and deletion, we insert one point at a time, and then delete one point at a time after all points are inserted. GRINCH implementation does not support batch insertion or deletion.
For Grove insertion, we batch insert 99\% of the points and then insert one point at a time.  We also run an experiment on MNIST by first inserting 1,000 points ($1.42\%$ of all points),
and then one point at a time.
We note that Grove does not support deletions.

\textbf{Evaluation.}
We evaluate the clustering quality using the Normalized Mutual Index (NMI).
NMI is MI normalized by the arithmetic mean of the entropy of the two clusterings. 
The NMI score is $1$ for a perfect correlation, and 0 for no mutual information.


For all algorithms except Grove, a flat clustering is extracted from the hierarchical clustering by cutting the dendrogram at a particular cutting threshold.
We try cutting at 40 log-spaced fixed thresholds between $10^{-4}$ and $0$ to find the best NMI with different cutting thresholds.
For Grove, we look at the clustering of all levels, and choose the one with highest NMI.

\vspace{-0.2cm}
\subsection{Comparing with Baselines}
Figures \ref{fig:ari} and \ref{fig:time} depict the NMI and clustering time of the algorithms. 

\textbf{Quality.} 
We show that \dynhac{} maintains a high quality dendrogram.
\Cref{fig:nmi_bar} shows the NMI of the clustering after all  insertions. 
\Cref{fig:nmi_mnist} shows the NMI of the algorithms after each update on MNIST. 


We see that \dynhac{} can get NMI very close to that of static HAC, which aligns with our theoretical approximation guarantee on the dendrogram quality, and matches prior experimental studies of static approximate HAC~\cite{parhac, terahac}. 
On ALOI, in the final dendrogram with all points inserted, the NMI obtained by \dynhac{} is only 0.0014 lower than the one of static HAC. 
For ILSVRC, \dynhac{} achieves even slightly higher NMI than the static HAC. Over all insertions and deletions on MNIST, \dynhac{}'s NMI is at most 0.03 and 0.015 lower than the static HAC, respectively. 

Comparing to \dynhac{}, GRINCH achieves a significantly lower NMI score. Compared to GRINCH, \dynhac{}'s NMI is 0.21 higher on MNIST, 0.08 higher on ALOI, and 0.12 higher on ILSVRC.
Grove achieves good NMI on ALOI and ILSVRC, but much lower NMI on MNIST.
Compared to Grove, \dynhac{}'s NMI 0.18 higher on MNIST, 0.001 higher on ALOI, and 0.007 higher on ILSVRC.

We conclude that \dynhac{} is the only method that we study that can consistently achieve quality close to that of the static HAC baseline.

\textbf{Running time.} 
In Figure~\ref{fig:time}(a) we show the running time of all algorithms.
Our algorithm is slower than GRINCH, which is expected if we compare the running time bounds.
Specifically, GRINCH's running time is $O(Tn + H^2)$ per data point, where $T$ is the time spent on nearest neighbor search, and H is the height of resulting dendrogram~\cite{monath2019scalable}.
On the other hand, in the worst case our algorithm has  $\Theta((m+n)R\log^2n)$ running time, where $R$ is the total number of rounds, i.e., in the case of an update resulting in all clustering merges to change. However, in practice we often do not need to update all merges, so our running time is still faster than the static HAC. Though our algorithm is slower than GRINCH and Grove, our clustering quality is higher as discussed above. 

On MNIST, ALOI, and ILSVRC, we are 423x, 41x, and 10.0x faster than the Static HAC for insertion, respectively. Deletion is 39.7x, 6.9x, and 1.56x faster on MNIST, ALOI, and ILSVRC. On ILSVRC, the total size of the dirty partitions is large, and so we have a smaller speedup.
\iffull
In the Appendix, 
\else
In the full version of our paper, 
\fi
we also show the running time for all insertions and deletion on the last 1\% of all three data sets.

\vspace{-0.2cm}
\section{Conclusion}
We introduce the first fully dynamic HAC algorithm for the popular average-linkage version of the problem, which maintains a $1+\epsilon$ approximate solution.
\dynhac{} handles each update up to 423x faster than recomputing the clustering from scratch. At the same time it achieves up to 0.21 higher NMI score than the state-of-the-art dynamic hierarchical clustering algorithms, which do not provably approximate HAC.

\bibliographystyle{IEEEtran}
\bibliography{references}

\clearpage

\iffull
\appendix
\section{Expanded Related Work}

Several studies have considered the problem of maintaining a hierarchical clustering with average-linkage by re-arranging the hierarchy of clusters in the presence of node insertions \cite{menon2019online, zhang1996birch, monath2023online}. 
These attempts can only process incremental updates. 
BIRCH \cite{zhang1996birch, garg2006pbirch} works for Euclidean spaces; it maintains a tree structure with each node storing statistics on the nodes in its subtree. When a new point arrives, it tracks a root-to-leaf path based on some closeness criteria, and eventually inserts the new point as a leaf in the tree.  
PERCH~\cite{kobren2017hierarchical} and GRINCH~\cite{monath2019scalable} work similarly to each other, where they first identify the leaf representing the nearest neighbor of the newly inserted node, and then track the path to the root of the hierarchy and apply appropriate rotations, as well as "graft" operations which are designed to discover chain-like clustering structures.
OHAC~\cite{menon2019online} processes the insertion of a node by first deleting all nodes in the path from the the nearest neighbor of a newly arrived node to the root decomposing the dendrogram to a forest and then re-runs HAC on these roots of the forest.

While the linkage function that has received the bulk of attention is the average-linkage function, the single-linkage -- the similarity between the clusters is the maximum of the-point-to-point similarities -- has also been studied in the dynamic setting.
In fact, the single-linkage function is closely related to the Minimum Spanning tree problem (see e.g.,~\cite{tseng2022parallel}), whose dynamic version has been extensively studied in the literature~\cite{holm2001poly, kapron2013dynamic, nanongkai2017dynamic}.

Another line of work considered optimization objectives for hierarchical clustering \cite{charikar2019hierarchical, vainstein2021hierarchical}. Some of the studied objectives include the CKMM Revenue \cite{cohen2019hierarchical}, Dasgupta cost~\cite{dasgupta2016cost}, and the MW Revenue \cite{moseley2023approximation}. That line of work is more theory-focused. In particular the best algorithms wrt Dasgupta objective give $O(\sqrt{\log n})$ approximation~\cite{charikar2017approximate, cohen2019hierarchical}, while conditional lower-bounds exclude the existence of constant factor approximate solutions \cite{charikar2017approximate}.
In the case of Moseley-Wang objective, a random clustering (clearly very weak from a practical standpoint) has been shown to give a $1/3$ approximation, while the best approximation algorithm obtains a $0.585$ approximation \cite{alon2020hierarchical}. 
For the CKMM objective the best known algorithm achieves a $0.74$ approximation guarantee \cite{naumov2021objective}.
Interestingly, HAC with average-linkage has been shown to give a $1/3$ approximation for the  MW Revenue objective, and $2/3$ for the CKMM objective, and these bounds are tight \cite{charikar2019hierarchical}.

Finally, \cite{rajagopalan2021hierarchical} uses a hyperplane partitioning method to construct a hierarchical clustering over a stream of updates in a top-down manner, and has the nice property that is agnostic to the order of the updates. The method is only applicable to Euclidean spaces.

In the more general context of dynamic clustering algorithms, multiple formulations of clustering have been studied both for metric spaces, like  facility-location~\cite{bhattacharya2022efficient, cohen2019fully} and variants of $k$-clustering problems~\cite{henzinger2020fully, braverman2017clustering, goranci2021fully, bateni2023optimal, lkacki2024fully, fichtenberger2013bico}, as well as on graphs like $k$-core decomposition~\cite{liu2022parallel}, and correlation clustering~\cite{cohen2022online, behnezhad2023single}.

\section{Pseudocode}

In Algorithm~\ref{alg:dirtypartition} we show our algorithm for identifying dirty partitions based on $\Delta_P$.

\begin{algorithm}
\caption{DirtyPartitions}\label{alg:dirtypartition}
\begin{algorithmic}[1]
\State \textbf{Input:} $\Delta_P:\{v : p \to p' | v \in V_{\text{new and neigh}} \}$, $G$
\State \textbf{Output:} $DP$ set of dirty partitions
\State $DP = \{\}$
\For{$\{v : p \to p'\} \in \Delta_P$} \Comment{$p$ might be $\emptyset$ indicating the vertex did not exist, and $\emptyset \notin G$.}
\State $DP$.Insert($\{p'\}$) \label{alg:dirtypartition:addp'}
\State $DP$.Insert($\{p\}$) if ($p \in G$ and $p$ is not blue)\label{alg:dirtypartition:addp}
\EndFor
\end{algorithmic}
\end{algorithm}

In \cref{alg:updatevmap}, we show the pseudocode for update vertex mapping $\texttt{VMap}$ and compute vertices to delete from next round.

\begin{algorithm}
\caption{UpdateVMap}\label{alg:updatevmap}
\begin{algorithmic}[1]
\State \textbf{Input:} $\mathcal{D}_p$, $p$, $V^d_i$, $\mathcal{D}_{\text{dirty}}$, \texttt{VMap}
\State \textbf{Output:} $V^d_{i+1}$
\State \textbf{Update:} \texttt{VMap}
\State $V_{\text{active}} \gets$ active nodes in $H_p$
\State $V^d_{i+1} \gets \emptyset$
\For{$v \in V_i^d$} 
\If{$\texttt{VMap}(v) \notin V^{\text{contracted}}_{\text{active}}$}
\State  $V^d_{i+1} \gets V^d_{i+1} \cup \texttt{VMap}(v)$ 
\EndIf
\State Remove $v$ from $\texttt{VMap}$
\EndFor

\For{$v \in V_{\text{active}}$} 
\State $r \gets \mathcal{D}_p$.root($v$)
\If{$\texttt{VMap}(v) \neq r$ }
\State  $V^d_{i+1} \gets V^d_{i+1} \cup \texttt{VMap}(v)$ 
\State $\texttt{VMap}(v) \gets r$
\EndIf
\EndFor
\State \textbf{Return} $V^d_{i+1}$
\end{algorithmic}
\end{algorithm}

\section{Analysis of \dynhac{}}
We now outline the analysis of \dynhac{}.
The following lemma shows that all merges within a partition that is {\em not} marked as dirty are still good.

\begin{restatable}{lemma}{turndirty}\label{lemma:turn_dirty}
If a partition $P$ exists and does not become dirty upon node update, all $(1+\epsilon)$-good merges within the partition are still $(1+\epsilon)$-good.
\end{restatable}
\begin{proof}
By Definition~\ref{def:good}, a $(1+\epsilon)$-good merge can stop being $(1+\epsilon)$-good if (1) $\wmax(u)$ increases, (2) $\minmerge(u)$ decreases,  or (3) $u$ is deleted (same for $v$). We show that if a partition does not become dirty, none of the three cases can happen for $u$. The same arguments can be made for $v$.

First, observe that $u$ must be still in $P$, otherwise $P$ will be identified as dirty partition by \cref{alg:dirtypartition:addp}.
We now analyze the three cases above.
(1) If $\wmax(u)$ increases, then $u$ must have a new neighbor with higher edge weight than its current maximum neighbor edge weight. However, by Definition~\ref{def:delta_p}, if $u$ is the neighbor of a new node, it will be included in $\Delta_P$, and $P$ will be identified as dirty partition by \cref{alg:dirtypartition:addp}.
(2) Only nodes with the same merge sequence can have the same node id. So $\minmerge(u)$ cannot change. 
(3) If $u \in P$ is deleted, $P$ would be identified as a dirty partition by \cref{alg:dirtypartition:addp}. Since we would have $u$ changing partition from $P$ to $\emptyset$. So $u$ must still in $P$.
\end{proof}
Using Lemma~\ref{lemma:turn_dirty}, we can show that all merges made by the \dynhac{} algorithm are $(1+\epsilon)$-good, which in turn implies that the algorithm computes a $(1+\epsilon)$-approximate dendrogram, yielding the next theorem.

\begin{restatable}{theorem}{thmapprox}
\dynhac{} maintains a $(1+\epsilon)$-approximate dendrogram upon node insertions and deletions.
\end{restatable}
\begin{proof}
All merges in \dynhac{} are made by \subhac{}. So all merges are $(1+\epsilon)$-good when the merge is made. 
A merge is only untouched after an update if the partition of both nodes in the merge still exists and is not dirty. 
By Lemma~\ref{lemma:turn_dirty}, good merges stay good if its partition is not dirty. 
If a partition does not exist anymore, its red node must have been deleted. 
So all its neighbors (which is all nodes in the removed partition) must belong to a dirty partition, and re-merged by \subhac{}. 
So all merges in \dynhac{} are $(1+\epsilon)$-good after updates.
By Lemma~\ref{lemma:dendrogram}, any dendrogram produced by a sequence of $(1+\epsilon)$-good merges is $(1+\epsilon)$ approximate. So \dynhac{} maintains a $(1+\epsilon)$-approximate dendrogram upon node insertions and deletions.
\end{proof}

Finally, the following two theorems bound the amount of work that our algorithms perform (1) during a single round, and (2) to initialize the data structure given an initial input.

\begin{restatable}{theorem}{thmtotaldirty}
The total size of dirty partitions in a round can be bounded by the size of the 4-hop neighborhood of all inserted and deleted nodes. 
\end{restatable}
\begin{proof}
Consider an inserted/deleted node $x$.
This update may cause the neighbor of $x$ to change its partition.
Assume that it partition id becomes $r \neq x$ (or was $r \neq x$ prior to an insertion).
The partition subgraph containing $r$ can contain at most the 2-hop neighbors of $r$ (including the inactive nodes). $r$ is $x$'s 2-hop neighbor. So in total, $x$ can make at most its 4-hop neighbor dirty.
\end{proof}

\begin{restatable}{theorem}{thminit}
Inserting $n$ nodes with $m$ edges into an empty graph using \cref{alg:dynhac} takes $O(R(m + n) \log^2 n)$, where $R$ is the number of rounds. The space complexity is $O(Rm)$.
\end{restatable}
\begin{proof}
The bottleneck of \cref{alg:dynhac} is \subhac{}. The running time of \subhac{} \shepchange{on a graph containing $n$ vertices and $m$ edges} is $O((m + n) \log^2 n)$~\cite{terahac}.
\end{proof}

\section{Deferred Proofs}\label{apx:proof}

\changepartition*
\begin{proof}
A vertex can have different partition id in two cases: 1) it does not exist in $G$, or 2) it has a different partition id. The set of vertices satisfying case $1$ is $V$.

Now we look at case 2. A vertex can only change the partition id if its neighborhood changes, so only the neighbors of $V$ and $V^d$ can change their ids. Each such vertex is contained either in $(N_{G}(V^d) \setminus V^d)$ (neighbors of deleted vertices that are themselves not deleted), or $N_{G'}(V)$. 
\end{proof}

\section{Performance Analysis}

\textbf{Varying $\epsilon$}
Figure~\ref{fig:epsilon_bar} shows the average running time and NMI, using different $\epsilon$ values ($[0, 0.1, 1]$). Note that $\epsilon=0$ is the same as exact HAC. Here we show the result for insertions, but the result for deletion is similar\iffull
(see Appendix). 
\else 
. We present the result in the full version of our paper.
\fi
We observe that higher $\epsilon$ results in faster running times, with only a little degradation in clustering quality. Specifically, in insertion \dynhac$_{\epsilon=0.1}$ is up to 1.93x and \dynhac$_{\epsilon=1}$ is up to 2.6x times faster than \dynhac$_{\epsilon=0}$. In deletion \dynhac$_{\epsilon=0.1}$ is up to 2.6x and \dynhac$_{\epsilon=1}$ is up to 4.22x times faster. 

\begin{figure}[t]
    \centering
    \includegraphics[trim={0 10.2cm 0 0},clip, width=0.8\columnwidth]{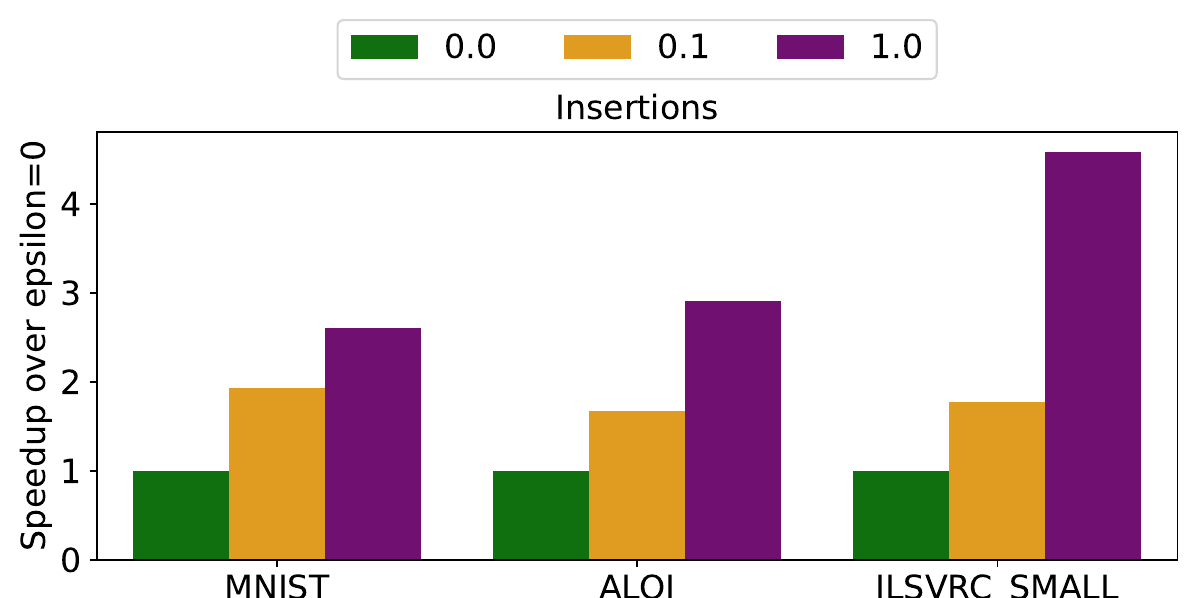}
    \includegraphics[trim={0 10.2cm 0 0},clip, width=0.8\columnwidth]{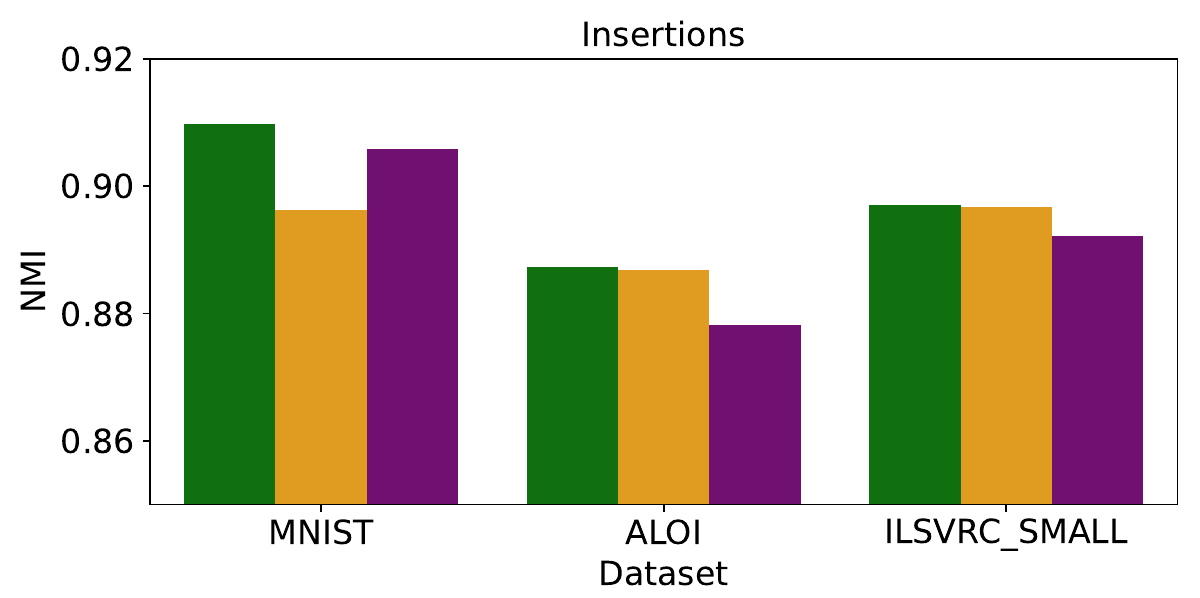}
    \caption{\small Update speedup over $\epsilon=0$ and NMI of the last 1\% insertions on data sets with different $\epsilon$ values. Deletions are similar. 
    }
    \label{fig:epsilon_bar}
    \vspace{-10pt}
\end{figure}

\textbf{Number of edges in dirty partition.} In \Cref{fig:num_edges}, we plot the number of edges in dirty partitions across all rounds against the update time. We see that the update time increases with the number of dirty edges, which aligns with our analysis that the bottle neck of the algorithm is \subhac{}, whose running time is  $O((m + n) \log^2 n)$~\cite{terahac}.

\textbf{Number of rounds.} 
In \cref{fig:num_rounds}, we plot the number of points inserted already and the number of rounds taken for an insertion update. 
We see that the number of rounds increases logarithmically with the of the number of points.

\begin{figure}[htbp]
    \centering
    \begin{subfigure}[b]{0.49\columnwidth}
        \centering
        \includegraphics[width=\linewidth, valign=c]{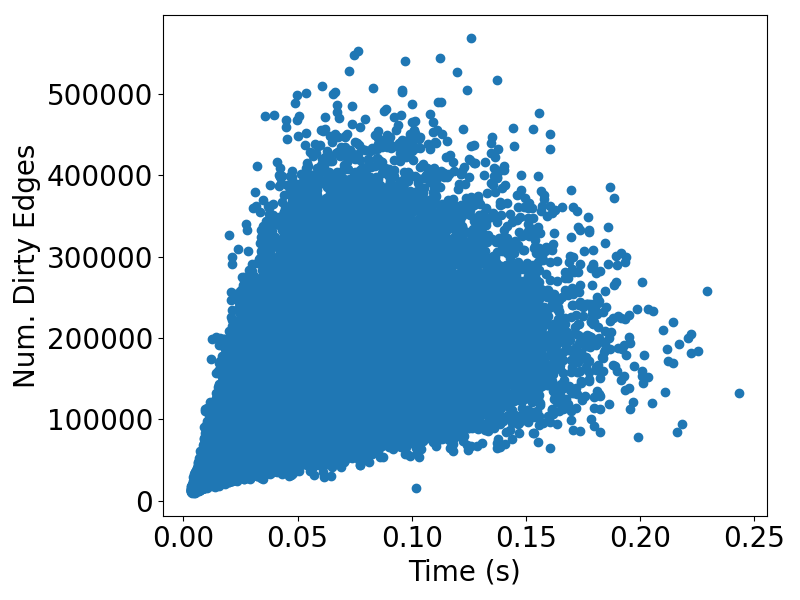}
        \caption{Total edges in all dirty partitions vs. update time.}
        \label{fig:num_edges}
    \end{subfigure}
    \begin{subfigure}[b]{0.49\columnwidth}
        \centering
        {\includegraphics[width=\linewidth, valign=c]{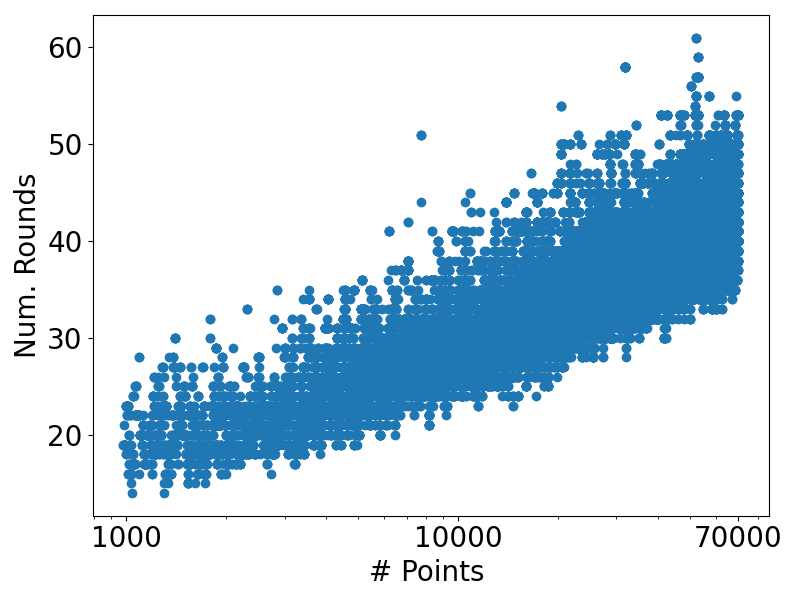}}
        \caption{The number of rounds vs. number of points inserted.}
        \label{fig:num_rounds}
    \end{subfigure}
    \caption{Analysis of \dynhac{} on MNIST.}
\end{figure}

\section{Additional plots and tables}

\begin{table}\footnotesize
\centering
\caption{\small Datasets, the number of data points $(n)$, the dimension $(d)$, and the number of ground truth clusters.}

\begin{tabular}[!t]{lrrr}   
\toprule
{Graph Dataset} & Num. Points & Dim. & Num. Clusters\\
\midrule
{\emph{MNIST    }}     & 70,000       & 2    & 10  \\
{\emph{ALOI    }  }     & 108,000       & 128    & 1000  \\
{\emph{ILSVRC   }}     &   50,000     & 2048    & 1000  \\
\vspace{-1cm}
\end{tabular}
\label{table:data}
\end{table}

\begin{figure*}
    \centering
    \includegraphics[width = 0.8 \columnwidth]{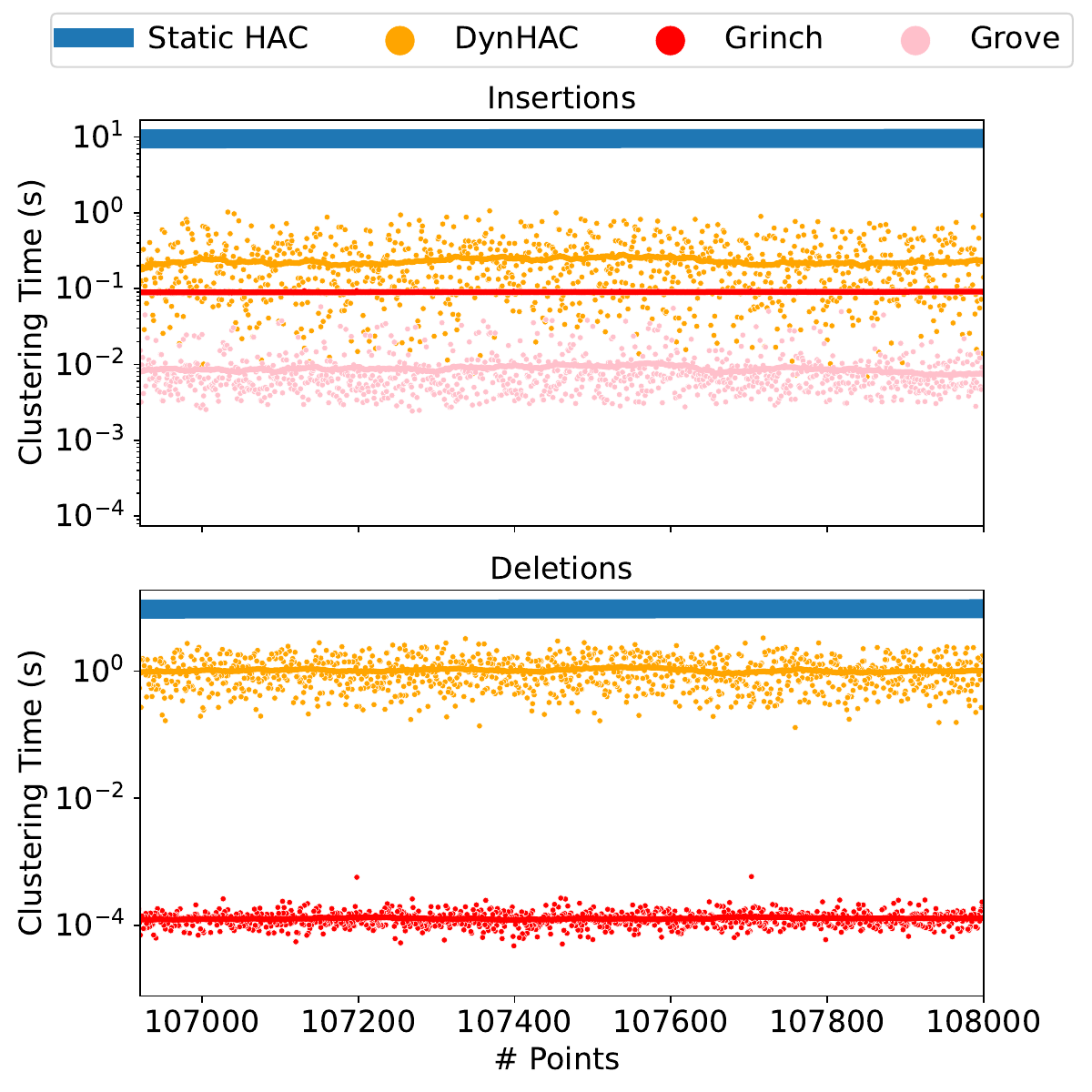}
    \includegraphics[width = 0.8 \columnwidth]{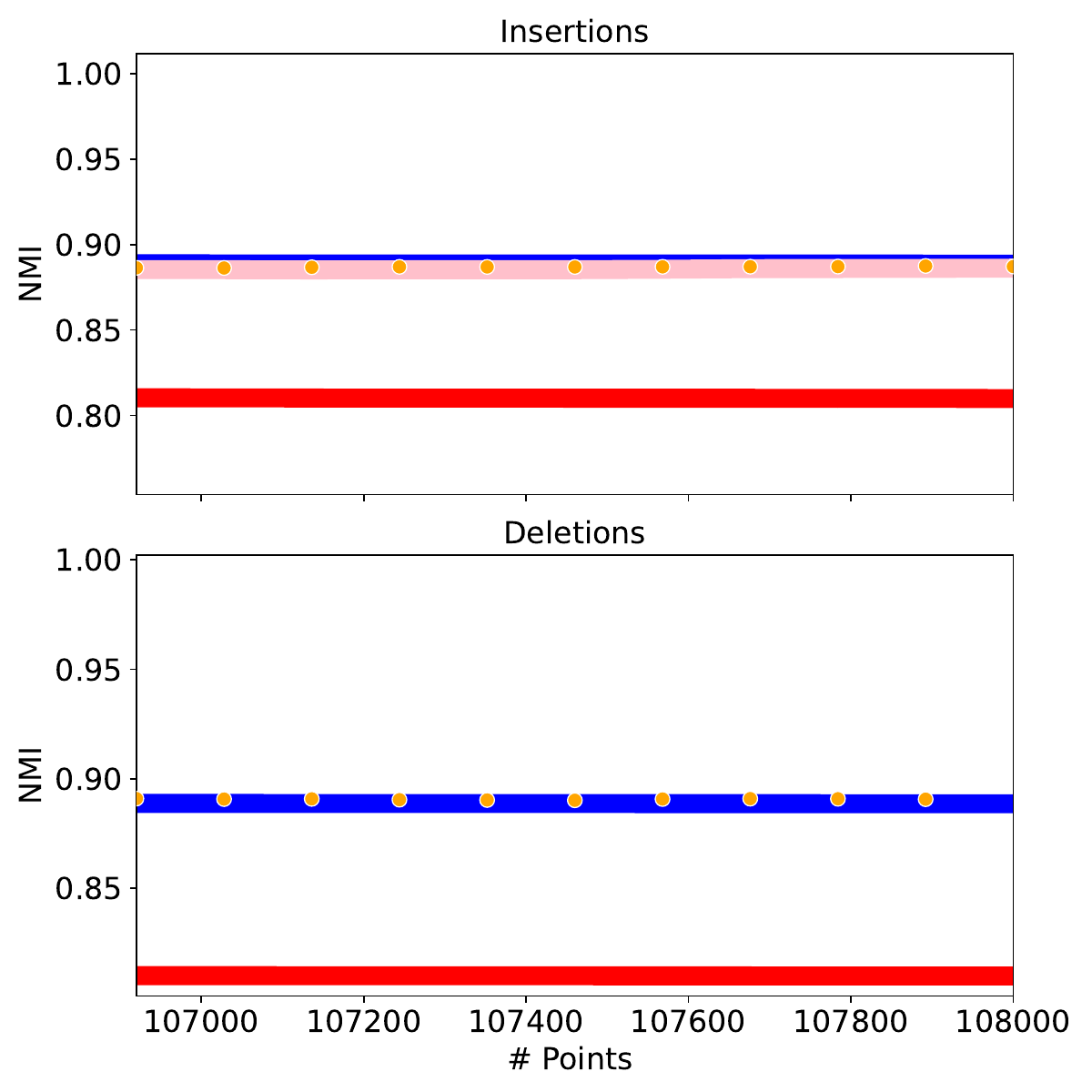}
    \caption{Running time and quality on ALOI for static HAC and our \dynhac{} insertion and deletion, and GINRCH insertion and deletion.}
    \label{fig:aloi}
\end{figure*}

\begin{figure*}
    \centering
    \includegraphics[width = 0.8 \columnwidth]{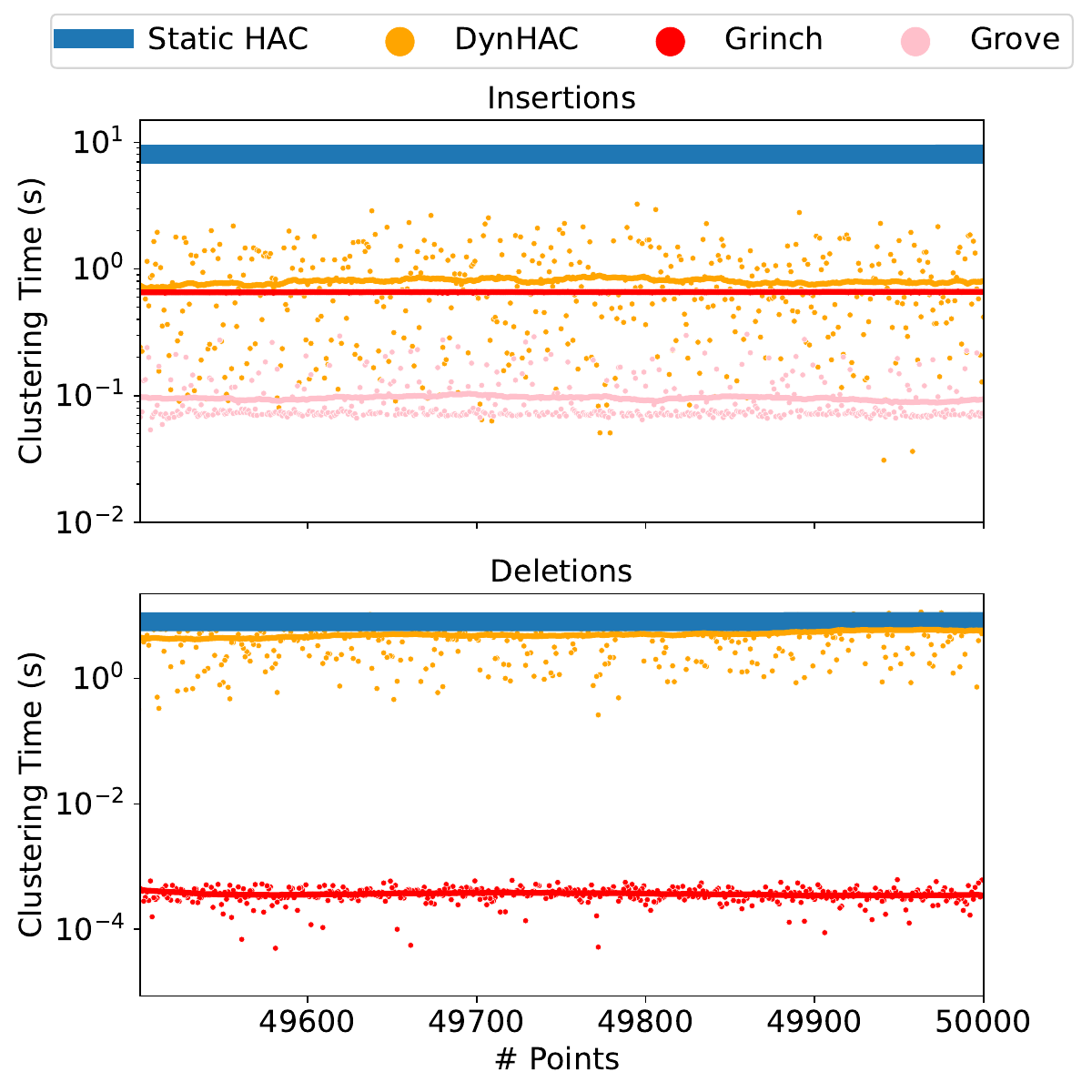}
    \includegraphics[width=0.8\columnwidth]{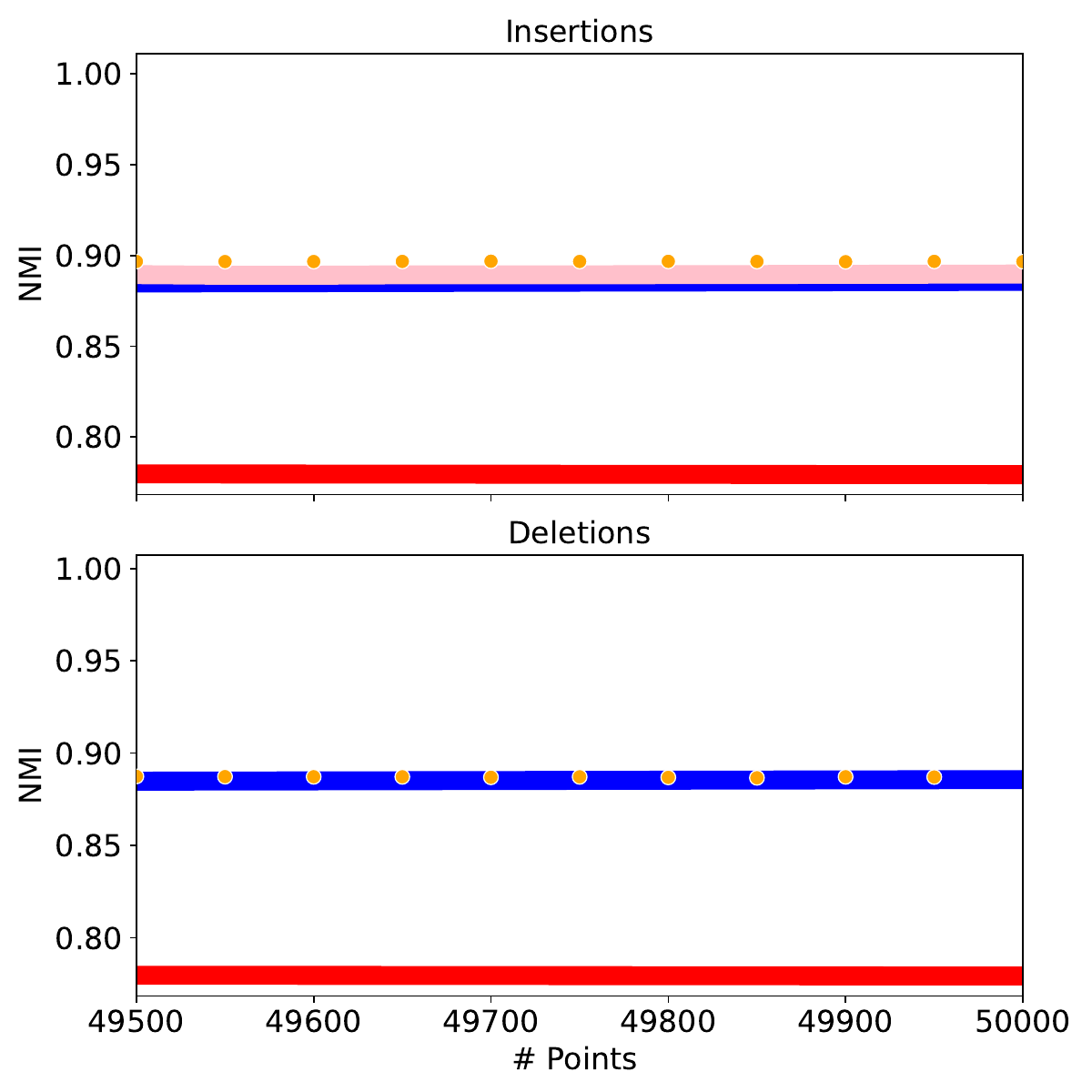}
    \caption{Running time and quality on ILSVRC for static HAC and our \dynhac{} insertion and deletion, and GINRCH insertion and deletion.}
    \label{fig:ilsvrc}
\end{figure*}

\begin{table*}[t]
\centering
\begin{tabular}{llrrr}
  \toprule
   & Algorithm & Clustering & NMI & Speedup \\
  Dataset &  &  &  &  \\
  \midrule
  MNIST & Static HAC & 3.687000 & 0.915277 & 1.000000 \\
  MNIST & DynHAC & 0.008703 & 0.900133 & 423.631747 \\
  MNIST & GRINCH & 0.001392 & 0.686119 & 2648.991842 \\
  MNIST & Grove & 0.002099 & 0.727114 & 1756.245833 \\
  \hline
  ALOI & Static HAC & 9.546000 & 0.888592 & 1.000000 \\
  ALOI & DynHAC & 0.230439 & 0.887153 & 41.425311 \\
  ALOI & GRINCH & 0.090949 & 0.809878 & 104.959618 \\
  ALOI & Grove & 0.007629 & 0.886217 & 1251.328324 \\
   \hline
  ILSVRC\_SMALL & Static HAC & 8.057000 & 0.885706 & 1.000000 \\
  ILSVRC\_SMALL & DynHAC & 0.804534 & 0.896586 & 10.014496 \\
  ILSVRC\_SMALL & GRINCH & 0.657990 & 0.779309 & 12.244861 \\
  ILSVRC\_SMALL & Grove & 0.093344 & 0.889749 & 86.315415 \\
  \bottomrule
  \end{tabular}

  \begin{tabular}{llrrr}
  \toprule
   & Algorithm & Clustering & NMI & Speedup \\
  Dataset &  &  &  &  \\
  \midrule
  MNIST & Static HAC & 3.687000 & 0.915277 & 1.000000 \\
  MNIST & DynHAC & 0.092865 & 0.900133 & 39.702794 \\
  MNIST & GRINCH & 0.000171 & 0.686119 & 21508.204239 \\
   \hline
  ALOI & Static HAC & 9.546000 & 0.888592 & 1.000000 \\
  ALOI & DynHAC & 1.382100 & 0.887153 & 6.906881 \\
  ALOI & GRINCH & 0.000200 & 0.809878 & 47722.081030 \\
   \hline
  ILSVRC\_SMALL & Static HAC & 8.057000 & 0.885706 & 1.000000 \\
  ILSVRC\_SMALL & DynHAC & 5.157170 & 0.896586 & 1.562291 \\
  ILSVRC\_SMALL & GRINCH & 0.000612 & 0.779309 & 13164.591869 \\
  \bottomrule
  \end{tabular}
    \caption{Clustering time averaged over the last (insertion) and first (deletion) 100 updates. NMI after the last insertion.}
    \label{tab:summary}
\end{table*}

\begin{figure*}
    \centering
    \includegraphics[width=0.7\columnwidth]{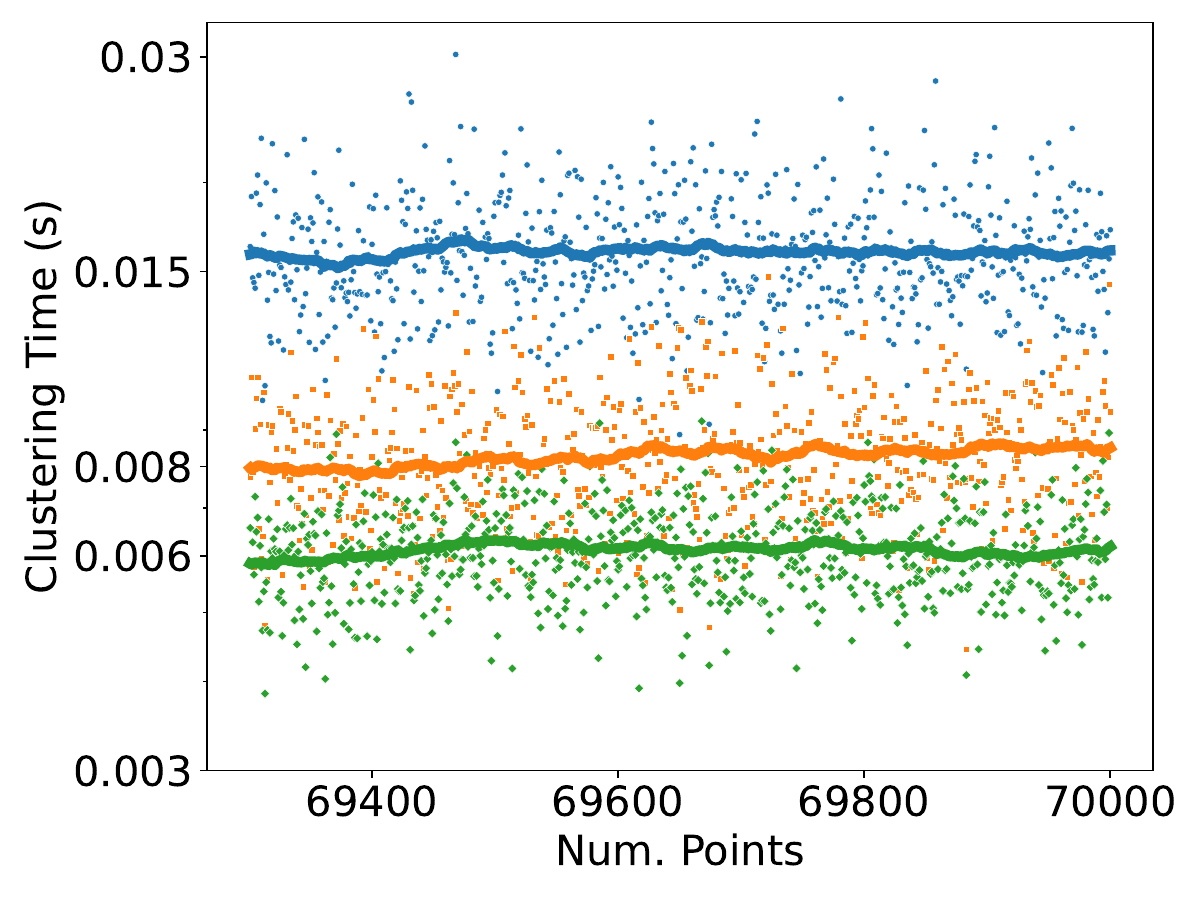}
    \includegraphics[width=0.7\columnwidth]{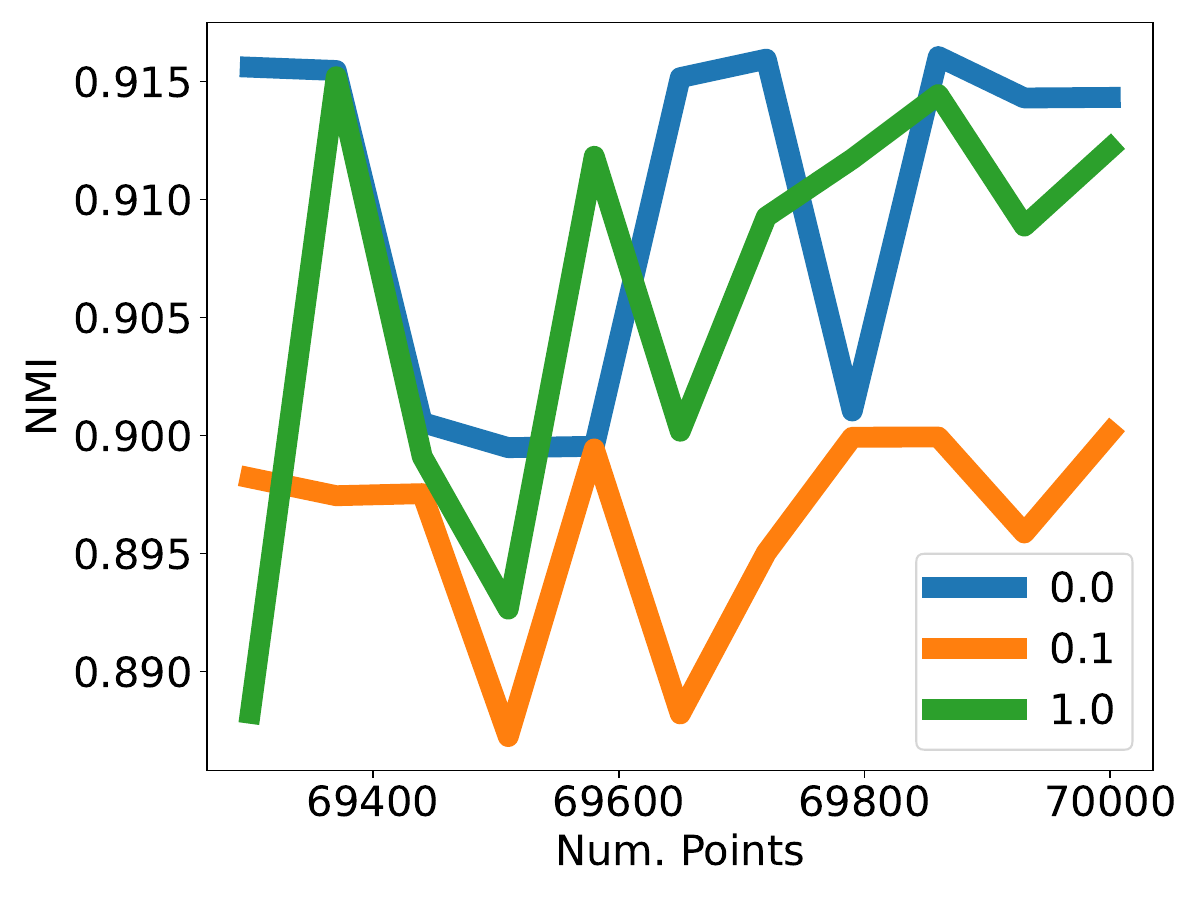}
    \caption{\dynhac{} Insertion with different $\epsilon$ values on MNIST.}
    \label{fig:mnist_ins_epsilon}
\end{figure*}

\begin{figure*}
    \centering
    \includegraphics[width=0.4 \textwidth]{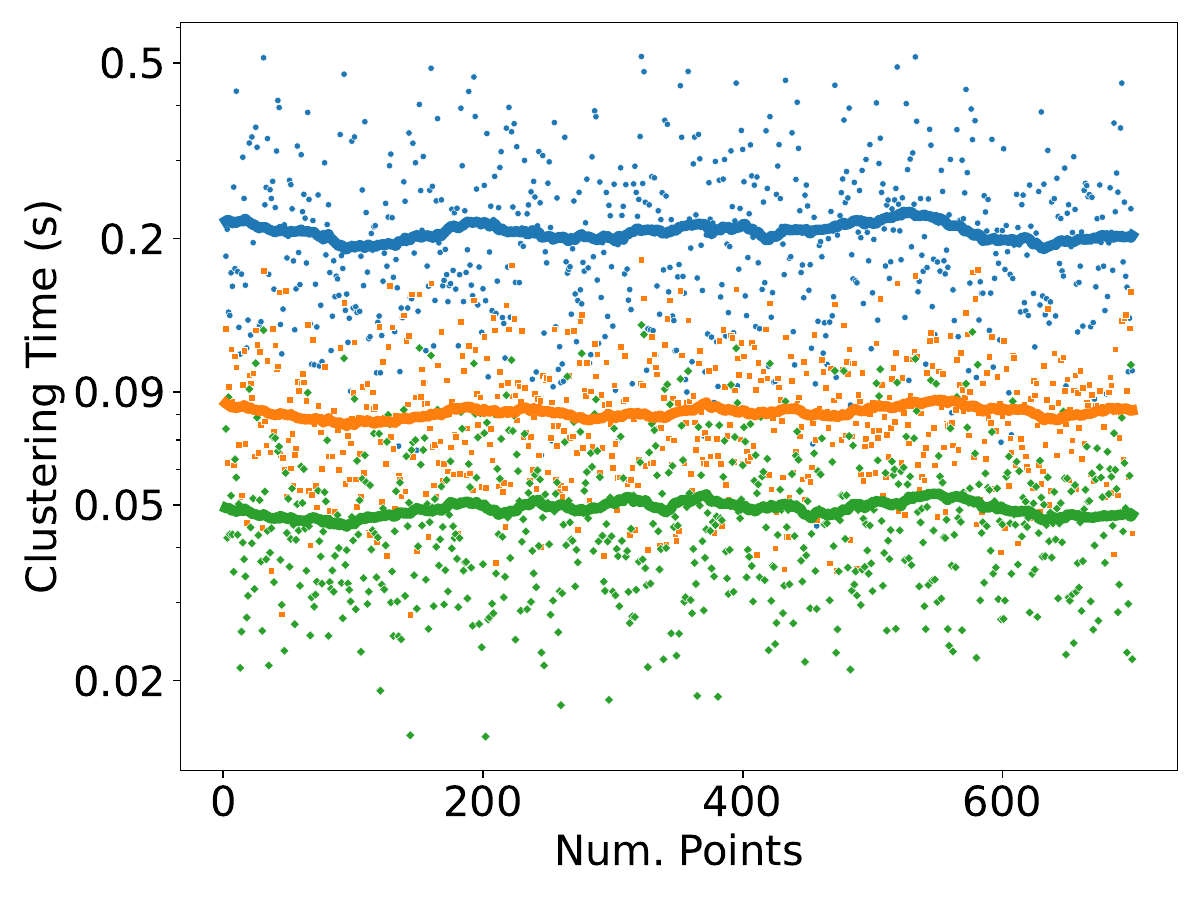}
    \includegraphics[width=0.4 \textwidth]{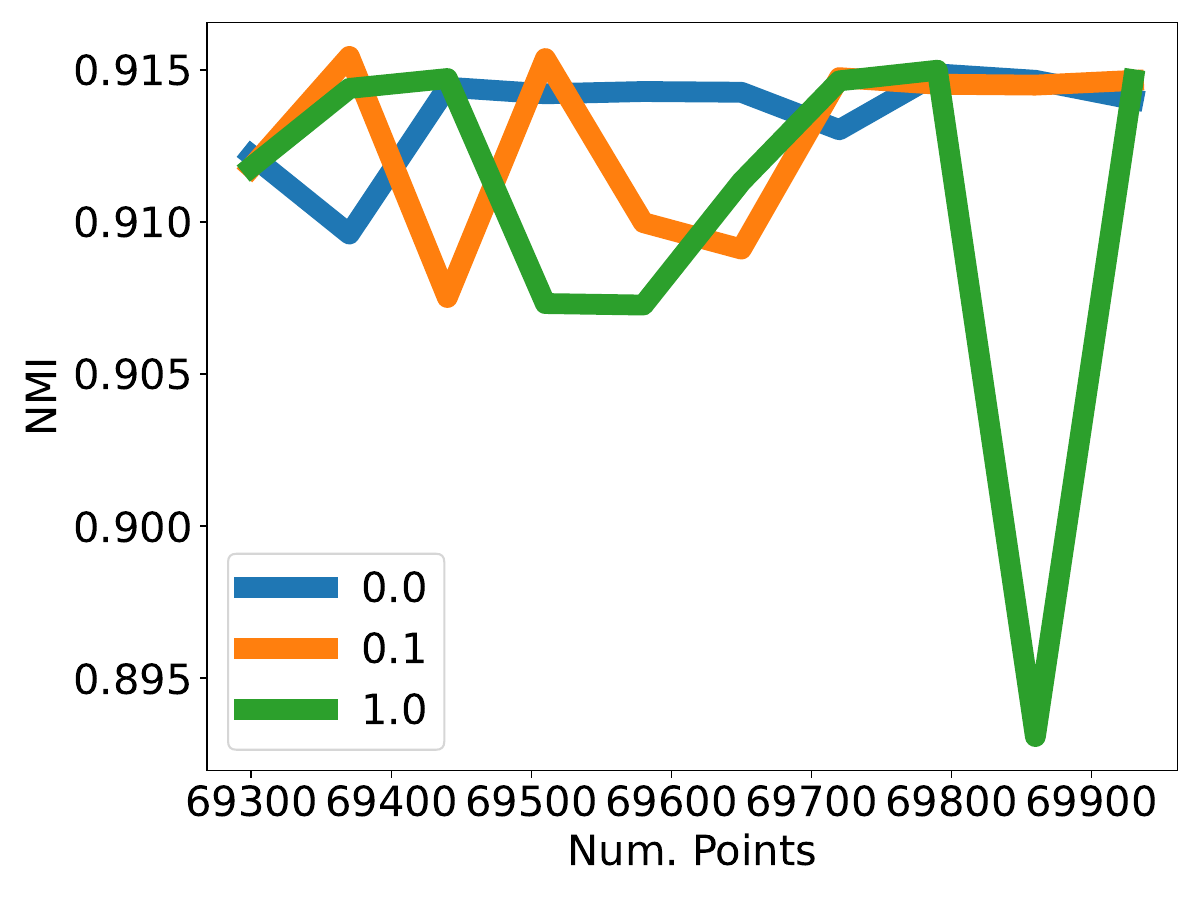}
    \caption{\dynhac{} Deletion with different $\epsilon$ values on MNIST.}
    \label{fig:mnist_del_epsilon}
\end{figure*}

\begin{figure*}
    \centering
    \includegraphics[width=0.4 \textwidth]{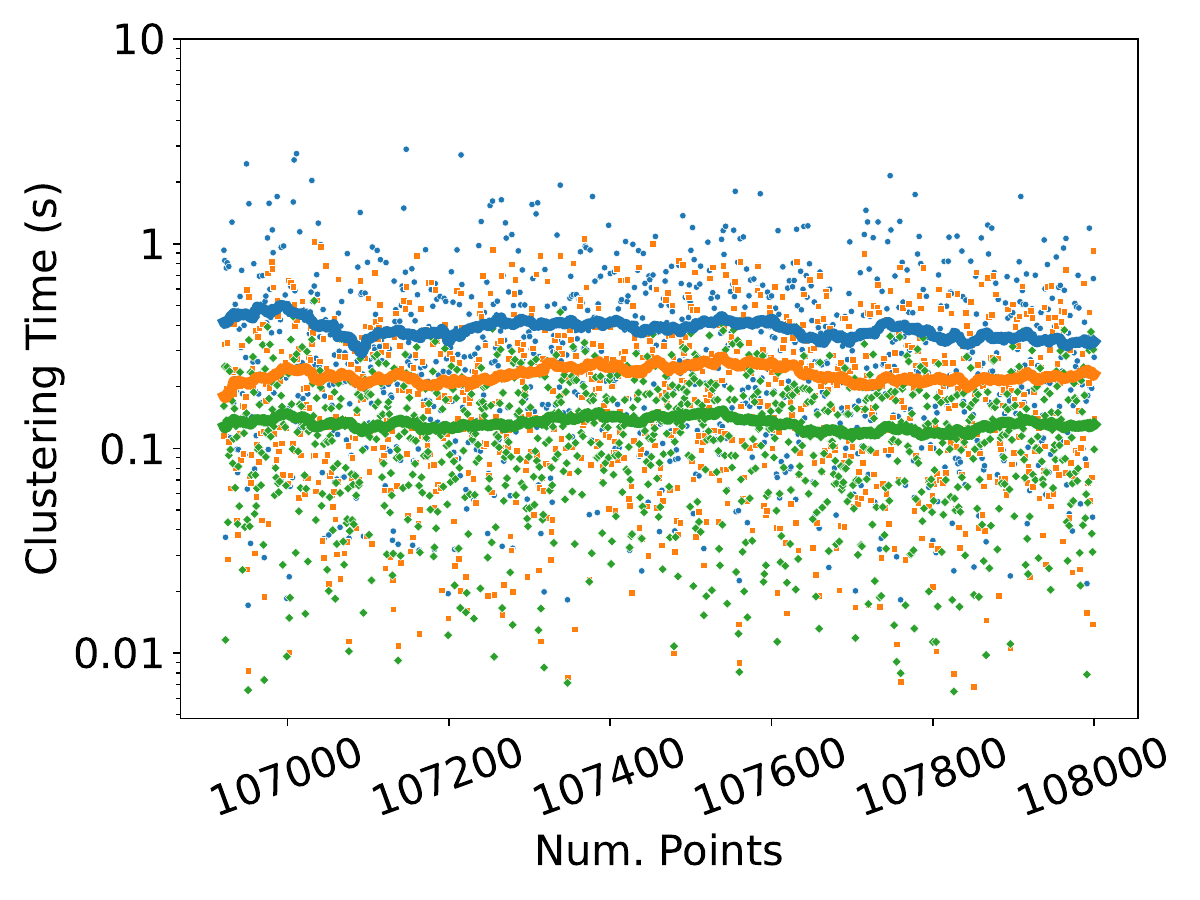}
    \includegraphics[width=0.4 \textwidth]{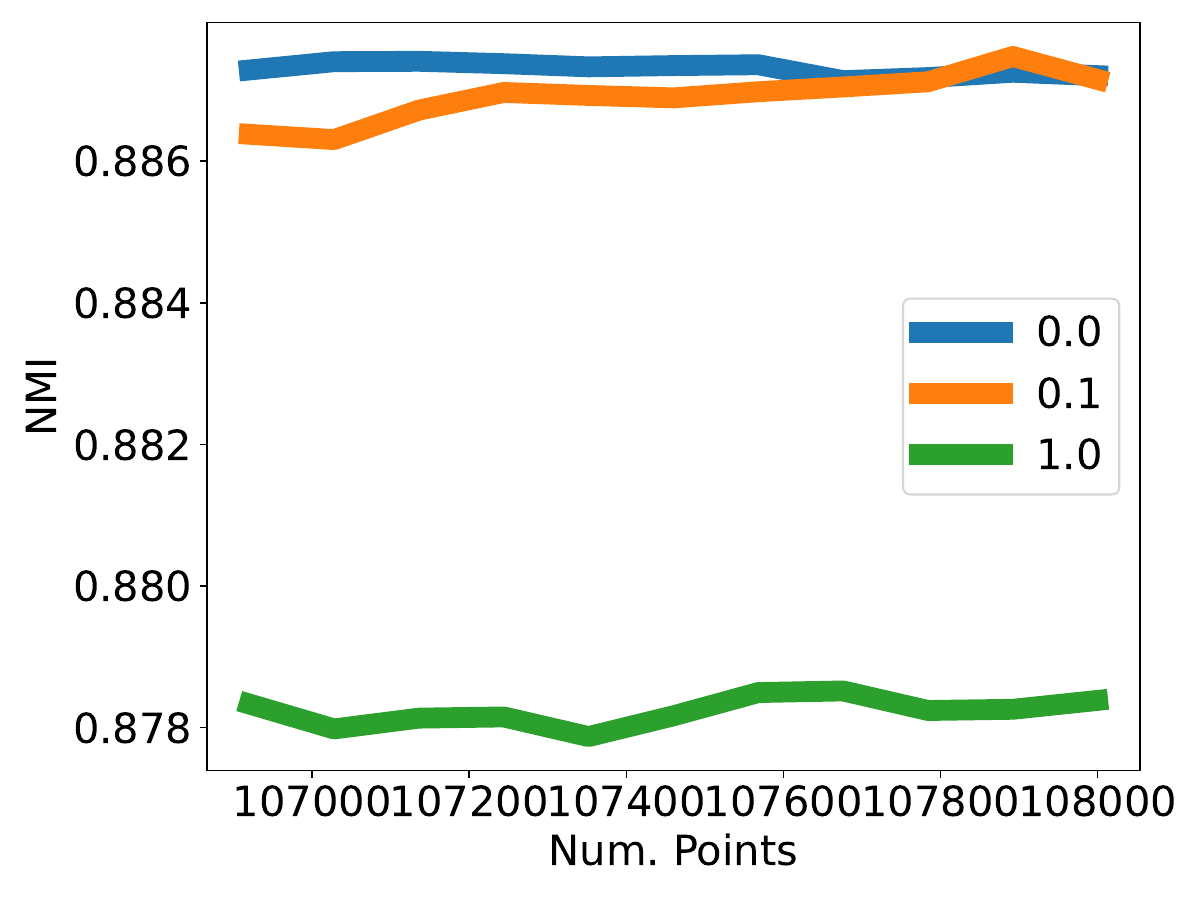}
    \caption{\dynhac{} Insertion with different $\epsilon$ values on ALOI.}
    \label{fig:aloi_epsilon}
\end{figure*}

\begin{figure*}
    \centering
    \includegraphics[width=0.4 \textwidth]{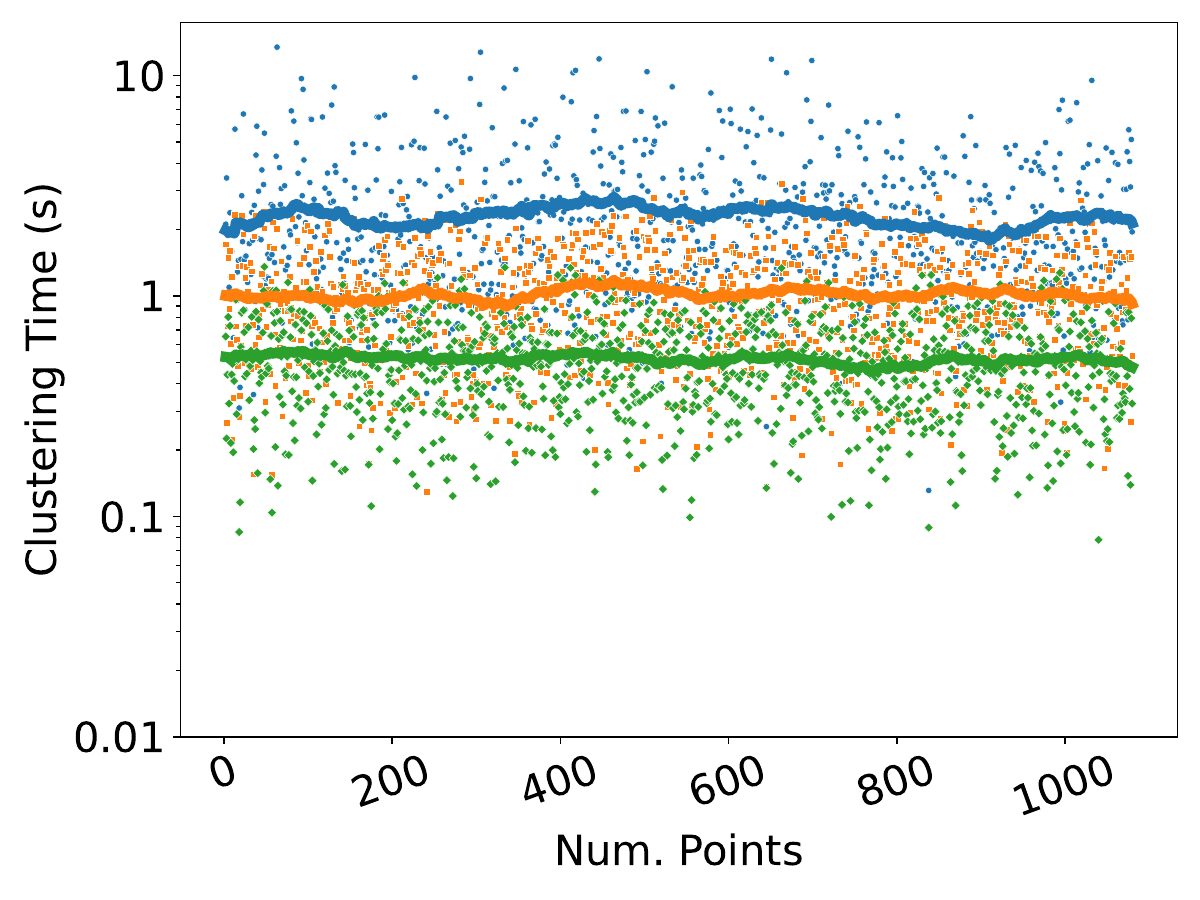}
    \includegraphics[width=0.4 \textwidth]{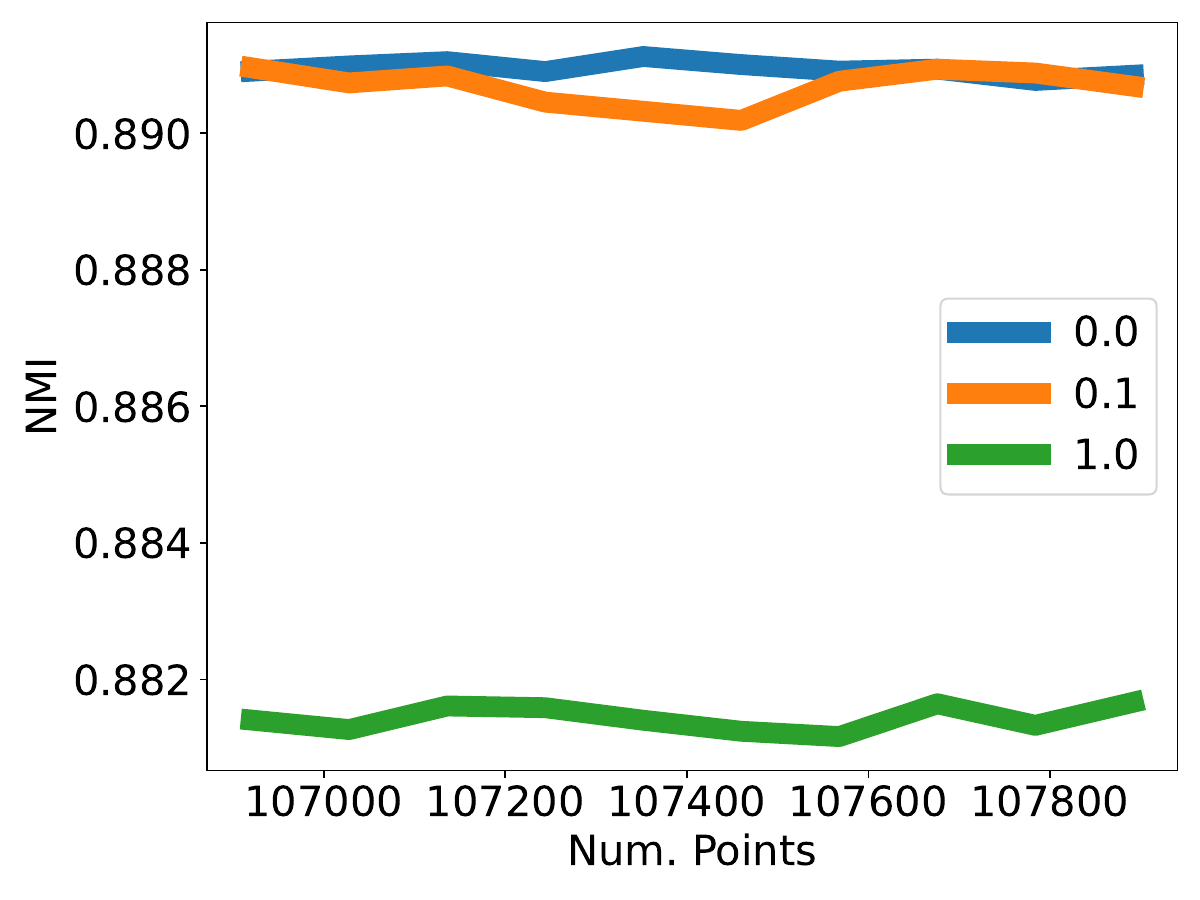}
    \caption{\dynhac{} Deletion with different $\epsilon$ values on ALOI.}
    \label{fig:aloi_epsilon}
\end{figure*}

\begin{figure*}
    \centering
    \includegraphics[width=0.4 \textwidth]{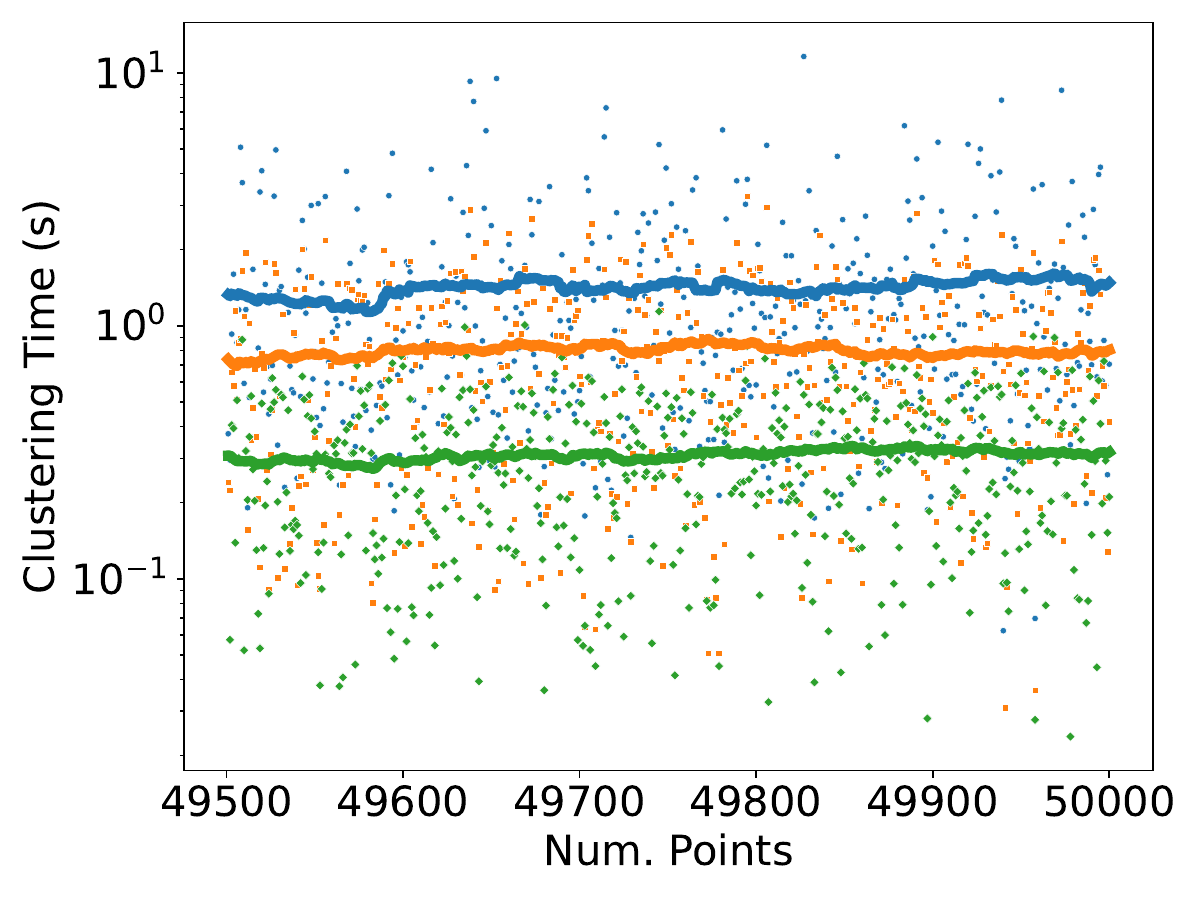}
    \includegraphics[width=0.4 \textwidth]{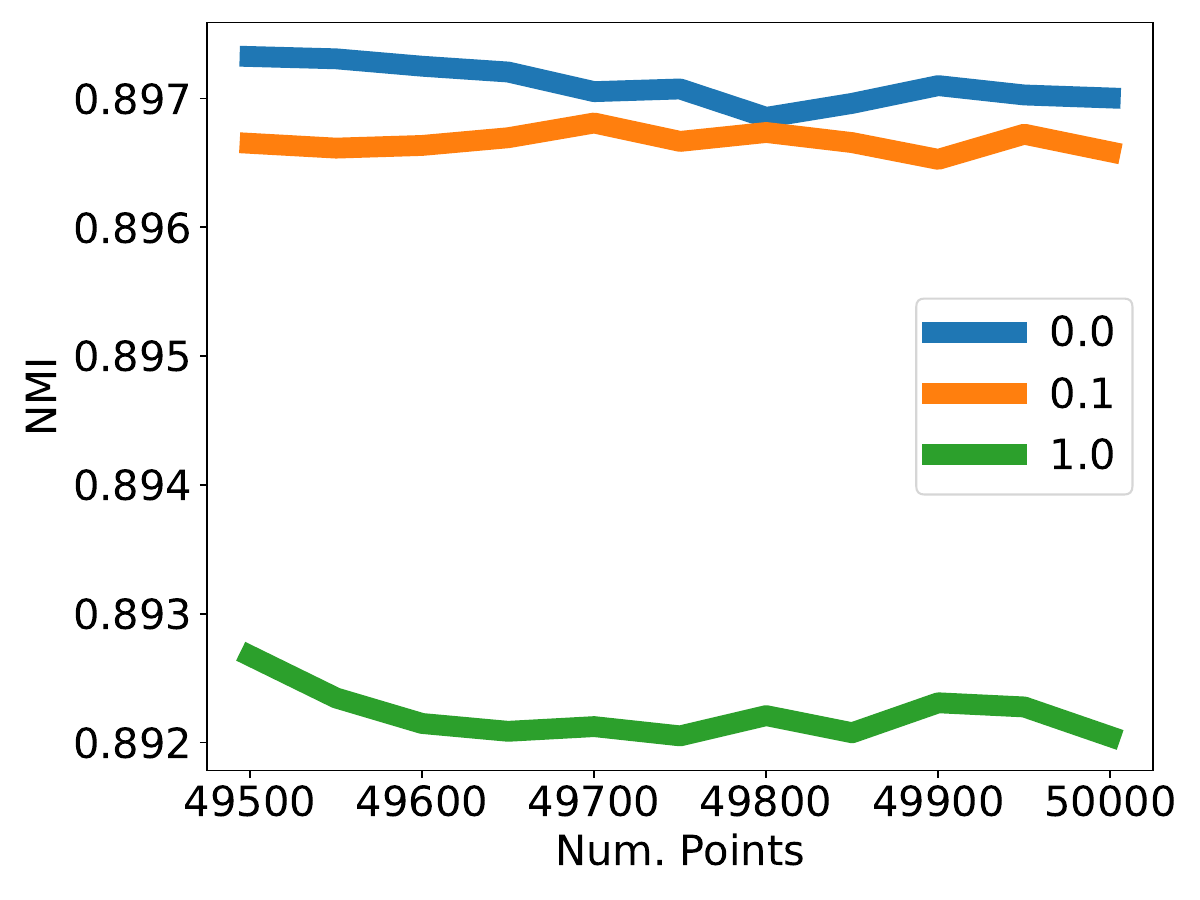}
    \caption{\dynhac{} Insertion with different $\epsilon$ values on ILSVRC.}
    \label{fig:ilvrc_ins_epsilon}
\end{figure*}

\begin{figure*}
    \centering
    \includegraphics[width=0.4 \textwidth]{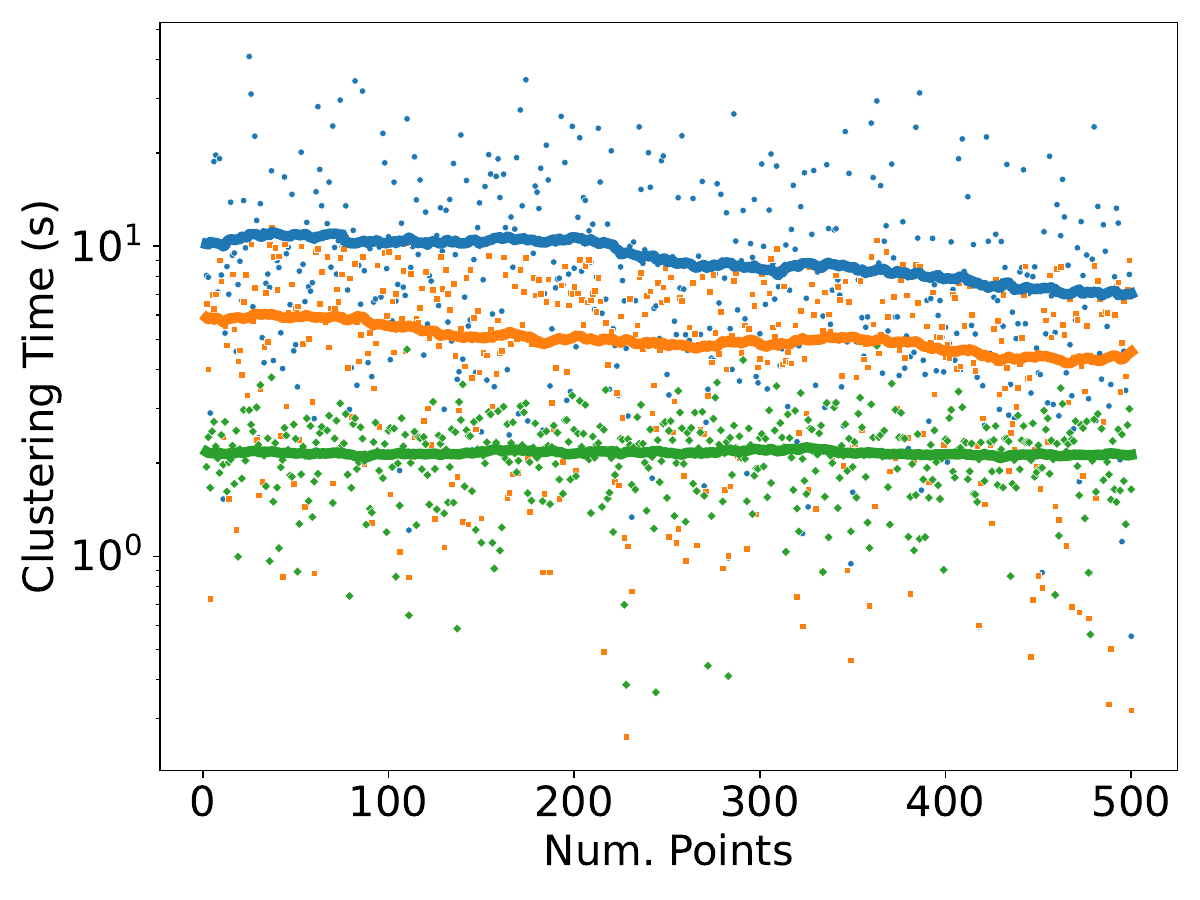}
    \includegraphics[width=0.4 \textwidth]{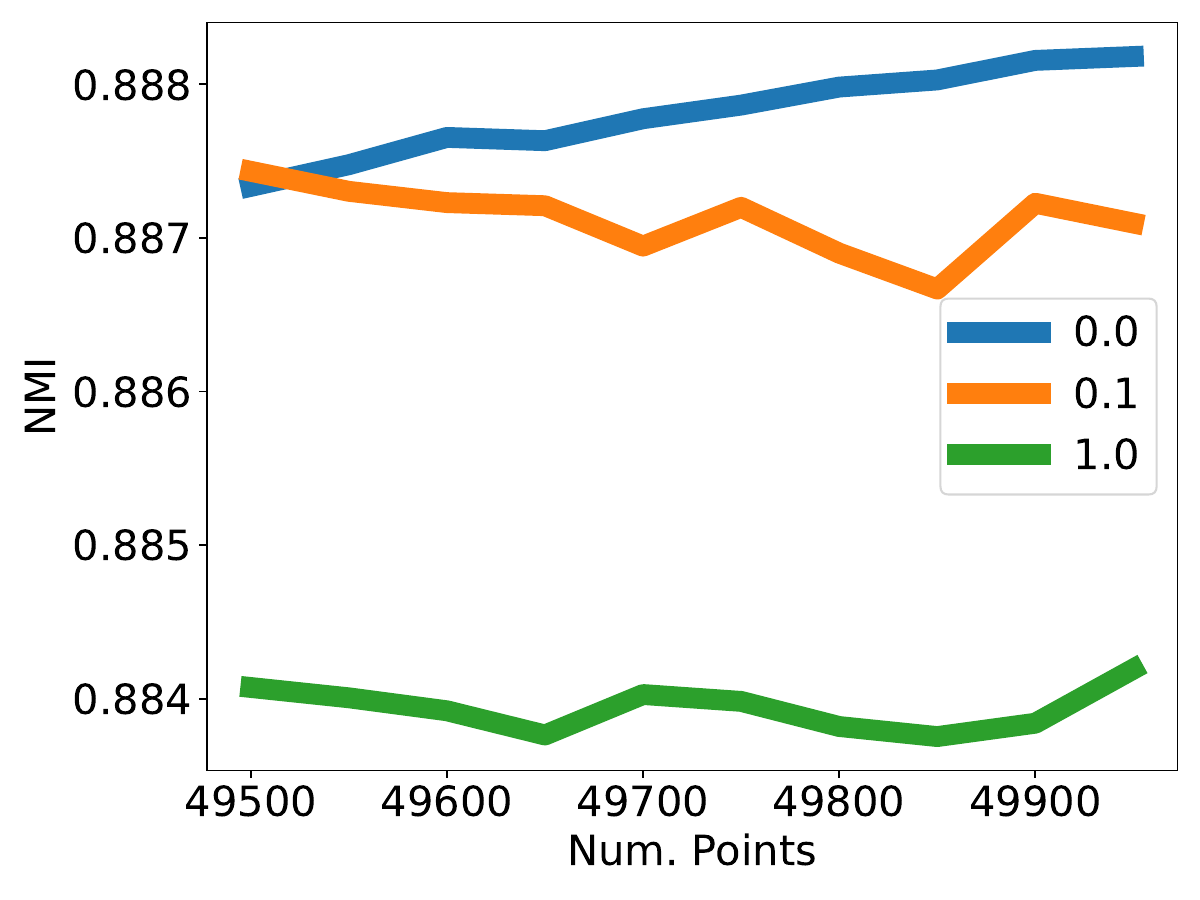}
    \caption{\dynhac{} Deletion with different $\epsilon$ values on ILSVRC.}
    \label{fig:ilvrc_del_epsilon}
\end{figure*}

\begin{table*}[]
    \centering
\begin{tabular}{llrrr}
\toprule
 & epsilon & Clustering & NMI & Speedup \\
Dataset &  &  &  &  \\
\midrule
MNIST & 0.0 & 0.015968 & 0.909773 & 1.000000 \\
MNIST & 0.1 & 0.008260 & 0.896266 & 1.933223 \\
MNIST & 1.0 & 0.006112 & 0.905788 & 2.612785 \\
ALOI & 0.0 & 0.383982 & 0.887300 & 1.000000 \\
ALOI & 0.1 & 0.229246 & 0.886907 & 1.674979 \\
ALOI & 1.0 & 0.132040 & 0.878230 & 2.908074 \\
ILSVRC\_SMALL & 0.0 & 1.413303 & 0.897105 & 1.000000 \\
ILSVRC\_SMALL & 0.1 & 0.796748 & 0.896663 & 1.773840 \\
ILSVRC\_SMALL & 1.0 & 0.308274 & 0.892214 & 4.584564 \\
\bottomrule
\end{tabular}

\begin{tabular}{llrrr}
\toprule
 & epsilon & Clustering & NMI & Speedup \\
Dataset &  &  &  &  \\
\midrule
MNIST & 0.0 & 0.206985 & 0.913558 & 1.000000 \\
MNIST & 0.1 & 0.080998 & 0.912764 & 2.555432 \\
MNIST & 1.0 & 0.049022 & 0.910421 & 4.222281 \\
ALOI & 0.0 & 2.292530 & 0.890944 & 1.000000 \\
ALOI & 0.1 & 1.017535 & 0.890677 & 2.253022 \\
ALOI & 1.0 & 0.517135 & 0.881430 & 4.433139 \\
ILSVRC\_SMALL & 0.0 & 9.178508 & 0.887809 & 1.000000 \\
ILSVRC\_SMALL & 0.1 & 5.057779 & 0.887121 & 1.814731 \\
ILSVRC\_SMALL & 1.0 & 2.154927 & 0.883938 & 4.259313 \\
\bottomrule
\end{tabular}
    \caption{Average running time and NMI of \dynhac{} with different values for $\epsilon$. Speedup is the speedup over $\epsilon=0$.}
    \label{tab:my_label}
\end{table*}

\begin{figure*}
    \centering
    \includegraphics[width=0.8\columnwidth]{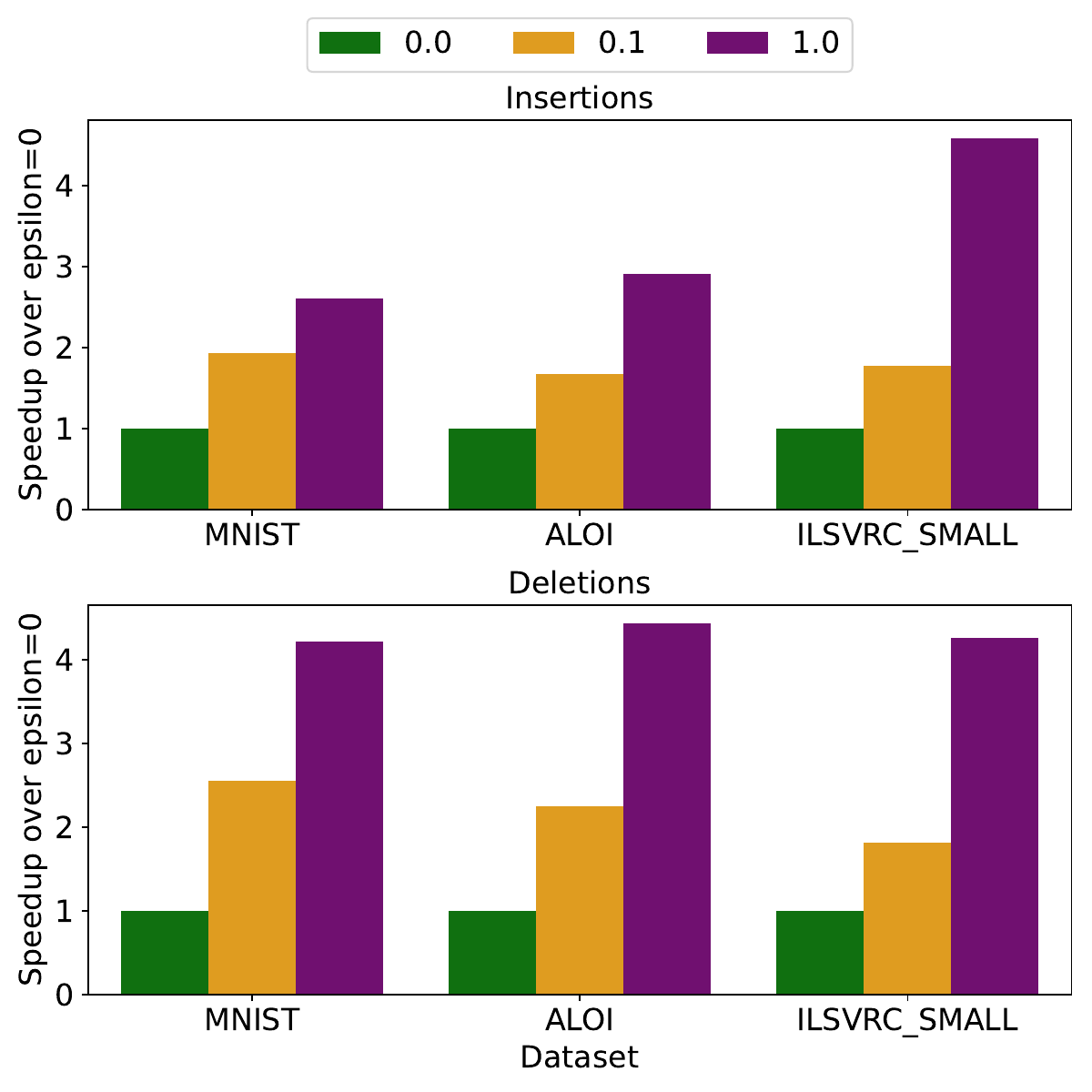}
    \includegraphics[width=0.8\columnwidth]{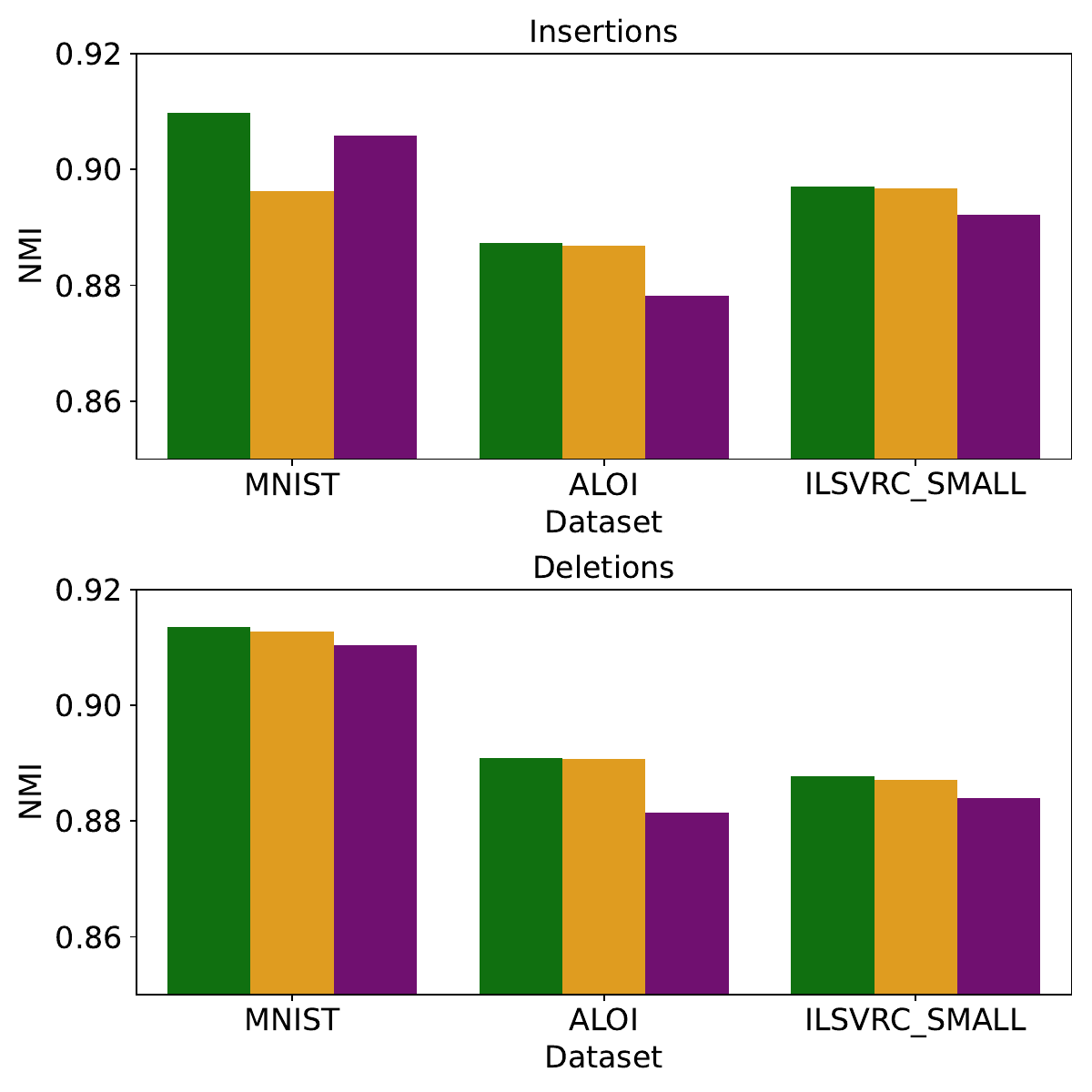}
    \caption{Speedup of \dynhac{} with different epsilon values over $\epsilon=0$, and the NMI values when using different epsilon values.}
    \label{fig:epsilon_bar_full}
\end{figure*}

\fi

\end{document}